\documentclass[10pt,reqno]{amsart}
\usepackage{amssymb,mathrsfs,color}
\usepackage{pinlabel}

\usepackage{amsmath, amsfonts, amsthm, verbatim, amssymb}
\usepackage{epstopdf, mathtools}
\mathtoolsset{showonlyrefs}
\usepackage{color}
\usepackage{esint}
\usepackage{stmaryrd}
\usepackage{booktabs}
\usepackage{graphicx}
\usepackage{graphicx}
\usepackage{subcaption}
\usepackage{placeins}

\usepackage{empheq}
\usepackage[shortlabels]{enumitem}

\newcommand{\rar}{\rightarrow}

\newcommand{\cl}{\mathcal}

\newcommand{\bs}[1]{\boldsymbol{#1}}

\usepackage{fullpage}

\newcommand{\eps}{\epsilon}

\newcommand{\p}{\partial}

\usepackage[multiple]{footmisc}
\numberwithin{equation}{section}

\newtheorem{thm}{Theorem}[section]
\newtheorem{cor}[thm]{Corollary}

\newtheorem{prop}[thm]{Proposition}

\theoremstyle{remark}
\newtheorem{rem}[thm]{Remark}

\usepackage{tensor}

\usepackage{graphicx}
\usepackage{xcolor} % A package to add color.
\usepackage{tensor}
\usepackage{slashed}
\usepackage{cite}

\newcommand{\dmsu}{\delta\msu}
\newcommand{\dmsv}{\delta\msv}
\newcommand{\vWpi}{\mathcal W_{\pi}}
\newcommand{\vWint}{\mathcal W_{\mathrm{int}}}
\newcommand{\vWinta}{\mathcal W_{\mathrm{int}}^{\rm a}}
\newcommand{\vWext}{\mathcal W_{\mathrm{ext}}}
\newcommand{\bbC}{\mathbb C}

\newcommand{\bbE}{\mathbb E}

\newcommand{\bbK}{\mathbb K}

\newcommand{\bbN}{\mathbb N}

\newcommand{\bbR}{\mathbb R}

\newcommand{\msm}{\bs{\mathsf{m}}}
\newcommand{\msn}{\bs{\mathsf{n}}}
\newcommand{\msv}{\bs{\mathsf{v}}}
\newcommand{\msa}{\bs{\mathsf{a}}}

\newcommand{\msg}{\bs{\mathsf{g}}}
\newcommand{\msh}{\bs{\mathsf{h}}}

\newcommand{\msu}{\bs{\mathsf{u}}}
\newcommand{\msx}{\bs{\mathsf{x}}}

\newcommand{\msX}{\bs{\mathsf{X}}}

\newcommand{\mse}{\bs{\mathsf{e}}}

\newcommand{\msE}{\bs{\mathsf{E}}}
\newcommand{\msQ}{\bs{\mathsf{Q}}}

\newcommand{\msT}{\bs{\mathsf{T}}}
\newcommand{\msY}{\bs{\mathsf{Y}}}
\newcommand{\msZ}{\bs{\mathsf{Z}}}

\allowdisplaybreaks[3]
\begin{document}

\title[Strain-gradient special Cosserat rods]{A direct theory of strain-gradient special Cosserat rods via the principle of virtual work}

\author{R. S. Hallyburton and C. Rodriguez$^*$}

\thanks{$^*$ denotes the corresponding author.}

\begin{abstract}
Motivated by recent advances in the synthesis of effective one-dimensional mechanical metamaterials, we develop a direct theory of strain-gradient special Cosserat rods based on the principle of virtual work. The internal virtual work is formulated intrinsically in terms of variations of the classical Cosserat strains and their spatial derivatives, without requiring a three-dimensional parent material or being restricted to elastic materials. Objectivity and insensitivity to rigid virtual displacements characterize the admissible internal-work density, while localization yields the dynamic field equations, higher-order contact interactions, natural boundary conditions, and balances of linear and angular momentum. For inextensible and unshearable rods, the constrained principle determines a unique multiplier functional, although its representation by reactive force and double-force fields possesses a gauge freedom absent from classical rod theory. We also formulate objective elastic, hyperelastic, and viscoelastic constitutive laws, extend transverse hemitropy, transverse isotropy, and flip symmetry to strain-gradient rods, and derive invariant representations of symmetric stored energies. As an application, we analyze a planar inextensible cantilever subjected to a compressive terminal force, with stored energy depending on the spatial derivative of curvature. We prove that the eigenvalues of the linearized problem are simple and coincide precisely with the loads at which bifurcation from the straight state occurs. In addition, each bifurcation point generates a connected global branch of nonlinear equilibria that is locally analytic near its origin.
\end{abstract} 

\maketitle

\section{Introduction}

\subsection{Higher-gradient one-dimensional continua}

Classical special Cosserat rod theories describe a slender body by a directed curve together with strain variables measuring stretch, shear, bending, and twist \cite{antman05}. Recent developments in the modeling and synthesis of mechanical metamaterials\footnote{Here, a mechanical metamaterial is understood as a material designed to achieve a prescribed macroscopic response through the organization of its microstructure. In this sense, a metamaterial is a material ``on demand,'' designed to satisfy prescribed functional requirements.}, however, have produced effective one-dimensional elastic continua whose stored energies also depend on spatial derivatives of the classical rod strains \cite{Barchiesi2019a,Barchiesi2020a,DellIsolaMoschiniTurco2024, MurciaTerranova2025,rizzi26,Boutin26}. These higher-gradient dependencies give rise to non-classical contact interactions and additional natural boundary conditions and, depending on the structure of the stored energy, may also permit non-rigid configurations that store no energy.

An important example is the \textit{pantographic} beam (see Figure \ref{fig:1}), introduced in the small-deformation setting in \cite{Alibert2003} and subsequently extended to finite deformations in \cite{Barchiesi2019a}. Beginning with a periodic structure composed of extensional and rotational springs, the authors apply a heuristic homogenization procedure in the spirit of Piola \cite{dellIsola2014b} to obtain an extensible, unshearable planar rod whose stored energy depends on both curvature and the gradient of stretch. The resulting energy controls the complete second gradient of the centerline, which motivates the terminology \emph{complete second-gradient rod}. Related pantographic constructions and their experimental realization have been studied in \cite{Barchiesi2020a}.

\begin{figure}[t]
    \centering
    \begin{minipage}{0.75\linewidth}
        \centering
        \includegraphics[width=\linewidth,trim=0 245.28 0 24,clip]{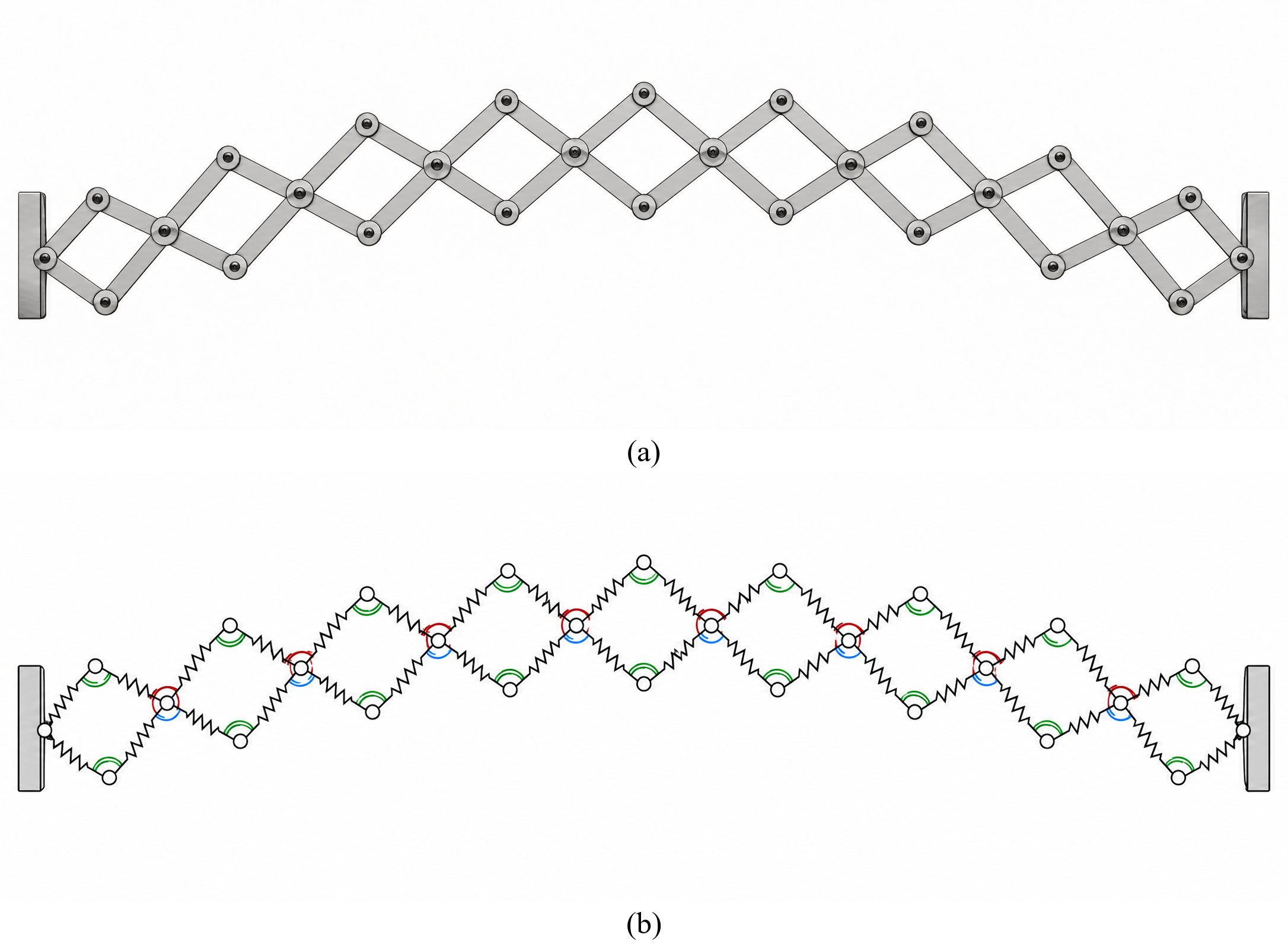}\par
        {\small (A)}\par\smallskip
        \includegraphics[width=\linewidth,trim=0 54.48 0 216,clip]{visuals/pantographic-beam-comparison.png}\par
        {\small (B)}
    \end{minipage}
    \caption{(A) A pantographic beam represented as a periodic lattice of interconnected microstructural cells. (B) The corresponding discrete mechanical model, with black zigzag lines representing extensional springs and colored arcs distinguishing rotational springs at the pivots. As the number of cells increases, the discrete structure is effectively described by a planar strain-gradient rod.}
    \label{fig:1}
\end{figure}

More recent synthesis procedures have led to effective rod models of still higher order. In \cite{rizzi26}, a periodic chain of modified Hart antiparallelograms is shown, through a heuristic homogenization procedure, to be effectively modeled by an inextensible, unshearable planar rod whose stored energy is a quadratic function of the derivative of curvature. Thus, every constant-curvature configuration has zero energy. The resulting rod therefore possesses \emph{floppy modes}, namely non-rigid configurations that store no energy. Other microstructures producing higher-gradient one-dimensional continua have been investigated in \cite{DellIsolaMoschiniTurco2024,Boutin26}; see also \cite{MurciaTerranova2025} for an overview of one-dimensional continua whose energies depend on the gradient of curvature. Both this model and the pantographic beam discussed in the preceding paragraph have stored energies that depend on the first spatial derivatives of strain variables appearing in classical special Cosserat rod theory. This common feature motivates the term \emph{strain-gradient special Cosserat rod} for the broader class of models developed in this work.

These examples belong to the broader \emph{synthesis problem} for mechanical metamaterials. Rather than beginning with a given material and determining its effective behavior, the synthesis problem begins with a desired mechanical effect and asks how to design a microstructure whose effective continuum response realizes that effect \cite{dellIsola2015b, dellIsola2023}. The desired response may be the ability to undergo large deformations with little or no energetic cost while remaining within the elastic regime. This behavior is enabled by floppy modes, which, at the continuum level, arise from the particular dependence of the stored energy on higher spatial gradients. For broader discussions of pantographic structures, mechanical metamaterials, and the synthesis problem, see
\cite{Barchiesi2019a, Barchiesi2019b,dellIsola2019b, dellIsola2019d,dellIsola2023, dellIsola2015b}.

Non-classical rod theories are important not only as effective one-dimensional descriptions of slender metamaterials, but also as building blocks for higher-dimensional models. Elastic surfaces reinforced by families of fibers have long been modeled as networks of interacting one-dimensional continua. For example, the work \cite{Steigmann2015a} developed a general theory of two-dimensional fibrous networks formed from interwoven families of spatial Kirchhoff rods \cite{antman05}. The surface energy accounts for fiber stretch and twist and generalizes earlier theories of fiber-reinforced sheets and networks going back to Rivlin \cite{Rivlin55, Adkins1955, WangPipkin86,WangPipkin1987, Steigmann1991}. This framework has motivated studies of exotic forms of buckling
\cite{Giorgio2015,Giorgio2016c}, the relation between fiber twist and surface geometry in the inextensible fiber setting \cite{Steigmann2018b}, and related phenomena. More generally, metamaterials may be constructed from lower-dimensional mechanical elements, or even from lower-dimensional metamaterials themselves
\cite{Alibert2003,Seppecher2011, Barchiesi2020a, RodriguezEtAl2026Tripantographic}. A sufficiently general theory of non-classical rods may therefore contribute to the synthesis of both one- and higher-dimensional continua.

The diversity of rod models produced by synthesis raises a basic theoretical question: how should one formulate a higher-gradient rod theory when no three-dimensional parent material is prescribed? Existing strain-gradient beam and rod models have largely been obtained by reducing three-dimensional strain-gradient theories. Several strain-gradient beam models \cite{Lazopoulos2014,Lazopoulos2010}, for example, combine three-dimensional strain-gradient elasticity with Bernoulli--Euler or Kirchhoff-type kinematic assumptions and then integrate over the cross section. These reduced theories have primarily been used to study linearized buckling and perturbative postbuckling behavior. More recently, {Yadav and Gupta} \cite{Gupta2024} derived a strain-gradient theory for special Cosserat rods through dimensional reduction of a three-dimensional strain-gradient theory of Toupin--Mindlin type
\cite{Toupin1964,Mindlin1964}. Their theory was used to investigate effects such as strain-gradient stiffening. Such reduction procedures are natural when the rod is understood as a slender specimen of a prescribed three-dimensional material. They are less suited to synthesized one-dimensional continua, whose effective constitutive response may have no evident three-dimensional precursor. This motivates a formulation based \textit{directly} on the intrinsic kinematics of a special Cosserat rod.

In this work, we develop such a formulation using the \emph{principle of virtual work}, following the approaches of Piola \cite{dellIsola2014b} and Germain \cite{Germain1973,Germain1973b, Germain2020}. At each fixed time, the internal virtual work is postulated directly in terms of variations of the classical special Cosserat strains and their spatial derivatives. Localization and integration by parts then determine the governing equations, the higher-order contact interactions, and the natural boundary conditions. This feature is particularly important in a higher-gradient theory. In general, classical balance laws determine the resultant force and couple acting on a material segment, but they do not, by themselves, identify the full hierarchy of higher-order contact interactions. Within the principle of virtual work, these interactions arise naturally as the quantities dual to the boundary values of the virtual displacement and its spatial derivative. Since the internal virtual work vanishes automatically on infinitesimal rigid displacements, the principle requires the external virtual work to vanish on such displacements; the classical balances of linear and angular momentum then follow.

We call the resulting theory a \emph{strain-gradient Cosserat rod theory} because the internal virtual work depends on variations of the first spatial derivatives of the classical Cosserat strain variables, which themselves depend on first derivatives of the centerline and director fields. The order may appear different after kinematic constraints are imposed. For example, in an inextensible and unshearable planar rod, the bending strain is the curvature of the centerline, and its spatial derivative involves the third derivative of the centerline. A second-gradient Cosserat rod may therefore reduce to what is called a third-gradient centerline or beam theory. This distinction reconciles the terminology used here with that of
\cite{DellIsolaMoschiniTurco2024,rizzi26,Boutin26,Eremeyev2026}.

Inspired by the rod energies arising from the synthesis literature, Eremeyev \cite{Eremeyev2026} studied the planar instability of third-gradient beams under several boundary conditions. By analyzing the corresponding linearized equations, he established quantitative relations between their critical forces and those of classical Euler--Bernoulli beams, showing that the additional gradient terms stiffen the rod. He also investigated postbuckling behavior through perturbation expansions. These calculations identify candidate critical loads and formally describe the nearby postbuckling response, but they do not \textit{rigorously} establish that the critical loads of the linearized problem are bifurcation points of the full nonlinear equilibrium problem.

To the best of our knowledge, a direct strain-gradient theory of special Cosserat rods based on the principle of virtual work has not previously been developed in a form that simultaneously develops governing equations, higher-order contact interactions, constitutive restrictions, and material symmetries. The first purpose of the present work is to develop such a theory. Unlike dimensional-reduction approaches, the theory is formulated directly from the intrinsic kinematics of a Cosserat rod and does not require the specification of a three-dimensional parent material. Moreover, the general theory is not constrained to elastic materials. Elastic, hyperelastic, and viscoelastic rods instead arise as constitutive specializations of the general framework. {Furthermore, our direct theory reproduces the governing equations, boundary conditions, and balance laws in the special case of a hyperelastic rod with a stored energy derived from a three-dimensional energy. This provides a substantial validation of our more general approach.}

The second purpose of the paper is to apply this theory to the planar buckling of an inextensible and unshearable strain-gradient Cosserat rod under a terminal compressive force. This problem is a higher-gradient analogue of the classical Euler elastica \cite{antman05}. We prove that bifurcation from the straight configuration occurs precisely at the eigenvalues of the associated linearized boundary value problem. Each bifurcation point belongs to a connected branch of nonlinear equilibria satisfying a global bifurcation alternative, and the branch is locally analytic near the bifurcation point. Thus, the critical loads identified by the linearized theory generate genuine branches of nonlinear equilibria, thereby providing a rigorous nonlinear justification for the predictions of Eremeyev \cite{Eremeyev2026}. This result provides a rigorous basis for local postbuckling expansions and illustrates how non-classical buckling phenomena can be analyzed within the intrinsic theory developed here.

\subsection{Main results and outline}

The paper consists of two main parts. The first develops a direct, dynamic theory of strain-gradient special Cosserat rods. The second applies this theory to a nonlinear planar buckling problem.

Section 2 reviews the kinematics of special Cosserat rods. A configuration is described by a centerline and an orthonormal frame of directors, or equivalently by a centerline and a rotation field. We recall the classical Cosserat strains and introduce their objective material representations. These strains and their spatial derivatives provide the natural arguments of the constitutive functions considered later.

{In Section~3, we formulate the principle of virtual work for strain-gradient special Cosserat rods. Objectivity and insensitivity to infinitesimal rigid virtual displacements characterize the most general internal-work density that is linear in the full second virtual jet: it is work-conjugate to variations of the classical Cosserat strains and their material derivatives. Starting from a virtual-work identity on the whole rod, we derive the bulk field equations, higher-order contact interactions, natural boundary conditions, and classical balances of linear and angular momentum. We also treat inextensible and unshearable rods and show that the constrained principle determines a unique multiplier functional, although its representation by individual reactive force and double-force fields is not unique, in contrast to the classical theory of Cosserat rods.}

Section 4 develops the objective constitutive theory for strain-gradient Cosserat rods, including elastic, hyperelastic, and viscoelastic response. For inextensible and unshearable rods, a heuristic penalty argument relates the reactive contact force and double force and selects a representative from their gauge-equivalence class. We then extend transverse hemitropy, transverse isotropy, and flip symmetry to the strain-gradient setting and derive invariant representations of the stored energy, including representations for energies defined away from an invariant closed set.

Section 5 considers an inextensible, unshearable, and untwisted planar strain-gradient rod that is clamped at one end and subjected to a compressive terminal force at the other. The equilibrium equation is a fourth-order nonlinear boundary value problem that may be viewed as a strain-gradient analogue of the Euler elastica. We first prove that the eigenvalues of the linearized problem are simple. After reformulating the nonlinear problem as a compact integral equation, we show that bifurcation from the trivial branch occurs precisely at these eigenvalues. From each bifurcation point, there emanates a connected branch that is either unbounded or contains another bifurcation point, and each branch is locally analytic near its point of origin. In particular, for the cantilever setting considered here, these results place the corresponding predictions of Eremeyev's purely linearized analysis on a rigorous nonlinear footing: each predicted critical load is realized as a bifurcation point of the full equilibrium equations. Finally, we provide numerical illustrations of the first three nontrivial buckling branches of the planar rod. {The important question of whether these branches intersect each other remains open}.

\subsection*{Acknowledgements} C. R. gratefully acknowledges support of NSF DMS-2307562. R. S. H. gratefully acknowledges support of the DoD NDSEG Fellowship. R. S. H. also acknowledges that this material is based upon work supported by the Air Force Office of Scientific Research under award number FA9550-25-C-B010 in the amount of \$133,200.

\subsection*{AI Disclosure} ChatGPT 5.6-Sol and Claude Opus 5.5 were used as adversarial readers, for literature review, to check for typos, and to improve the exposition. The authors retain exclusive responsibility for the correctness, accuracy, and validity of the results presented in this work.

\section{Preliminaries}

We briefly review the kinematics of special Cosserat rods (see Ch. 8 of \cite{antman05} for a more comprehensive introduction).\footnote{Unless otherwise noted, we use Einstein index notation; repeated indices denote summation. Greek and Latin indices range over $\{1, 2\}$ and $\{1, 2, 3\}$, respectively. }

\subsection{Kinematics and strains}
We denote three-dimensional Euclidean space by $\bbE^3$, and we identify its translation space with $\bbR^3$ via a fixed orthonormal basis $\{ \bs e_k \}_{k = 1}^3$. 
Throughout this work, we use standard vector and tensor operations in $\bbR^3$, and unless stated otherwise, repeated indices are summed over their range. 

Let $[0,L]$ be the reference interval parameterizing the material points of an initially straight rod, with reference length $L$.\footnote{The formulation presented herein admits a straightforward generalization to initially curved rods. To simplify the exposition, however, we consider only initially straight rods.} The \emph{current configuration} of a special Cosserat rod at time $t$ is defined by a triple: 
\begin{align}
[0,L] \ni s \mapsto (\bs r(s,t), \bs d_1(s,t), \bs d_2(s,t)) \in \bbE^3 \times \bbR^3 \times \bbR^3,
\end{align}
with $\bs d_1(s,t)$ and $\bs d_2(s,t)$ orthonormal for each $(s,t)$. The curve $\bs r(\cdot,t)$ is the \emph{centerline} of the rod at time $t$, and $\{\bs d_1(s,t), \bs d_2(s,t)\}$ are the \emph{directors} at $(s,t)$, vectors regarded as tangent to the current material cross section at $\bs r(s,t)$. 

Let
\begin{align*}
	\bs d_3(s,t) = \bs d_1(s,t) \times \bs d_2(s,t).
\end{align*}
Then $\{ \bs d_k(s,t) \}$ is a right-handed orthonormal basis for $\bbR^3$ describing a configuration of the material cross section at $\bs r(s,t)$. 
%\begin{figure}[htbp]
 %   \centering
  %  \includegraphics[width=0.96\linewidth]{visuals/cosserat_rod_configurations.pdf}
  %  \caption{Reference and current configurations of a special Cosserat rod. The Carolina blue curves are the centerlines, while the black vectors show the reference basis $\{\bs e_k\}$ and the current director frame $\{\bs d_k\}$. In the straight reference configuration, $\bs e_3$ is tangent to the centerline. In the current configuration, $\bs d_3$ need not be tangent; the displayed tilt illustrates shear.}
   % \label{fig:rod-configurations}
%\end{figure}
The rotation  
\begin{align}
	\bs R(s,t) &= \bs d_k(s,t) \otimes \bs e_k
\end{align}
is the unique rotation satisfying
\begin{align}
	\bs R(s,t) \bs e_k = \bs d_k(s,t), \quad k = 1,2,3. \label{eq:rotdir}
\end{align} 
Conversely, a rotation $\bs R(s,t)$ for each $(s,t)$ specifies directors $\{ \bs d_k(s,t)\}$ via \eqref{eq:rotdir} and a configuration at time $t$ can be equivalently defined in terms of the pair 
\begin{align}
	[0,L] \ni s \mapsto (\bs r(s,t), \bs R(s,t)) \in \bbE^3 \times \mathrm{Orth}^+. 
\end{align}

Since, for each $(s,t)$, $\bs R(s,t)$ is a rotation of $\bbR^3$, there exist unique vector fields
\begin{align*}
	\bs u(s,t) &= u_k(s,t) \bs d_k(s,t), \\ \bs w(s,t) &= w_k(s,t) \bs d_k(s,t),
\end{align*} 
such that for all $(s,t)$,
\begin{align}
\p_s \bs R(s,t) &= (\bs u(s,t)\times) \bs R(s,t), \label{eq:dsR}
\\ \p_t \bs R(s,t) &= (\bs w(s,t) \times ) \bs R(s,t), \label{eq:dtR}
\end{align}
where $\bs a \times$ is the skew-symmetric tensor generated by the axial vector $\bs a$. 
The components $u_1$ and $u_2$ of $\bs u$ are referred to as the \emph{flexural strains}, and the component $u_3$ is referred to as the \emph{torsional strain} (or \emph{twist}). The vector $\bs w(s,t)$ is referred to as the resultant angular velocity of the material section at $\bs r(s,t)$.

The relations \eqref{eq:dsR}-\eqref{eq:dtR} are equivalent to 
\begin{align}
	\p_s \bs d_k(s,t) = \bs u(s,t) \times \bs d_k(s,t), \quad 
	\p_t \bs d_k(s,t) = \bs w(s,t) \times \bs d_k(s,t), \label{eq:dd}
\end{align}
i.e., 
\begin{align}
	\bs u(s,t) = \frac{1}{2} \bs d_k(s,t) \times \p_s \bs d_k(s,t), \quad \bs w(s,t) = \frac{1}{2} \bs d_k(s,t) \times \p_t \bs d_k(s,t). \label{eq:uwexplicit} 
\end{align}
By differentiating \eqref{eq:uwexplicit} and using $\p_{st} \bs d_k(s,t) = \p_{ts} \bs d_k(s,t)$, we obtain the compatibility equation,
\begin{align}
	\p_s \bs w(s,t) = \p_t \bs u(s,t) + \bs u(s,t) \times \bs w(s,t). \label{eq:compatibility}
\end{align}
Given $\bs u$ (or $\bs w$), we may solve \eqref{eq:dsR} (or \eqref{eq:dtR}) and recover the directors via defining $\bs d_k(s,t) := \bs R(s,t) \bs e_k$.

We may also express the tangent vector to the {centerline} in the basis $\{ \bs d_k(s,t) \}_k$ via 
\begin{align}
	\bs v(s,t) := \p_s \bs r(s,t) = v_k(s,t) \bs d_k(s,t).
\end{align}
The components $v_1$ and $v_2$ are referred to as the \emph{shear strains}. The component $v_3$ is referred to as the \emph{dilatation strain}, and an orientation of the director $\bs d_3(s,t)$ relative to the {centerline} is fixed by requiring that for all $(s,t) \in [0,L] \times [0,\infty)$, 
\begin{align}
 v_3(s,t) > 0. \label{eq:v3pos}
\end{align} 
The restriction \eqref{eq:v3pos} also implies that the \emph{stretch} of the rod is never zero, $|\p_s \bs r(s,t)| > 0$, and that the rod cannot be sheared so severely that a section becomes tangent to the {centerline}. 

We define the \textit{pullbacks} of the vector fields $\bs u$, $\bs v$ via 
\begin{align}
\msu(s,t) := \bs R^T(s,t) \bs u(s,t) = u_k(s,t) \bs e_k, \quad \msv(s,t) := \bs R^T(s,t) \bs v(s,t) = v_k(s,t) \bs e_k.   
\end{align}
We make the simple but important observation that the vector fields $\msu(\cdot, t)$ and $\msv(\cdot, t)$ are \textit{objective}: if $\bs Q(t) \in \mathrm{Orth}^+$ and $\bs c(t) \in \bbR^3$, then for the motion 
\begin{align}
\label{eq:change_of_frame}
(\bs r^*(\cdot, t), \bs R^*(\cdot, t)) = (\bs c(t) + \bs Q(t)(\bs r(\cdot, t) - \bs o), \bs Q(t)\bs R(\cdot, t))
\end{align}
related to the motion $(\bs r(\cdot,t), \bs R(\cdot, t))$ by a change of frame, we have 
\begin{align}
 \msu^*(\cdot,t) = \msu(\cdot, t), \quad \msv^*(\cdot, t) = \msv(\cdot, t). \label{eq:uvobjective}
\end{align}

\section{Governing equations and balance laws via the principle of virtual work}

This section develops the governing equations of strain-gradient Cosserat rods from the principle of virtual work. We first characterize the most general local internal virtual-work density that is linear in the second virtual jet, objective, and vanishes on infinitesimal rigid displacements. We then show that the virtual-work identity on the whole rod determines the bulk field equations, the physical-end conditions, and compatible contact interactions on internal cuts, thereby inducing a balance of virtual work on every material subsegment. After deriving the classical balances of linear and angular momentum by Piola's method, we extend the formulation to internally constrained rods and characterize the resulting reactive fields and their gauge freedom, the latter property standing in stark contrast to the classical Cosserat rod theory.

\subsection{Principle of virtual work}

Fix a time $t$. Throughout this section we suppress the dependence of the fields on $(s,t)$ whenever no confusion can arise. In the unconstrained theory developed first in this section, the configuration manifold that we adopt is, for a fixed $p\in[1,\infty)$,
\begin{align}
	\cl M
	:=
	\left\{
	(\bs r,\bs R)\in
	W^{2,p}([0,L];\bbE^3\times\mathrm{Orth}^+)
	\ \middle|\
	v_3(s)>0\text{ for every }s\in[0,L]
	\right\}.
\end{align}
The strict positivity condition is meaningful pointwise because the relevant one-dimensional Sobolev fields have continuous representatives. {The constrained configuration spaces introduced later are treated as submanifolds of this unconstrained manifold; their tangent spaces therefore impose the corresponding restrictions on virtual displacements.} If $\mathcal N$ is a smooth manifold embedded in $\mathbb R^m$, we use the notation
\begin{align}
W^{k,p}([0,L];\mathcal N)
:=
\left\{
u\in W^{k,p}([0,L];\mathbb R^m)
\mid
u(s)\in\mathcal N \text{ for a.e. } s\in[0,L]
\right\}.
\end{align}
{The choice of $W^{2,p}$ is dictated by the strain-gradient kinematics. The internal virtual work involves variations of first derivatives of the classical strains with respect to the material coordinate $s$ and hence second derivatives of the configuration with respect to $s$. In the particular case of an \emph{unshearable} rod, discussed in Sections 3.4 and 5, the tangent to the centerline is aligned with the third director, and the internal virtual work depends on variations of centerline derivatives with respect to $s$ of order as high as three. A tangent vector to $\cl M$ at $(\bs r,\bs R)$ is interpreted as a generalized \emph{virtual displacement}. In the unconstrained theory developed in this section, its translational and rotational parts can be chosen independently; when internal constraints are imposed, they must instead belong to the corresponding tangent space.}

More precisely, let $(\bs r_\eps,\bs R_\eps):(-\eps_0,\eps_0)\to\cl M$ be a smooth curve through $(\bs r,\bs R)$ such that
\begin{align}
	\bs r_\eps = \bs r + \eps \bs \varphi + O (\eps^2), \quad \bs R_\eps = \exp(\eps \bs \psi \times) \bs R, \label{eq:virtvel}
\end{align}
{where $\bs\varphi,\bs\psi\in W^{2,p}([0,L];\mathbb R^3)$. Here $\bs\varphi$ is the virtual centerline displacement and $\bs\psi$ is the axial vector of the virtual rotation. Differentiating the objective strain fields with respect to the variation parameter gives}
\begin{align}
	\dmsu &:= \frac{d}{d\eps} \msu_\eps \Big |_{\eps = 0} = \bs R^T \partial_s \bs \psi, \label{eq:dotu}\\
	\dmsv &:= \frac{d}{d\eps} \msv_\eps \Big |_{\eps = 0} = \bs R^T \bigl (\partial_s \bs \varphi + \bs v \times \bs \psi \bigr ). \label{eq:dotv}
\end{align}

{We now formulate two requirements for an internal virtual-work functional directly at the level of the local internal virtual-work density. Since $\partial_s\dmsu$ and $\partial_s\dmsv$ involve second derivatives of $\bs\psi$ and $\bs\varphi$, respectively, define the full second jet of a virtual displacement by}
\begin{align}
	j_s^2(\bs\varphi,\bs\psi)
	:=
	\bigl(
	\bs\varphi,\bs\psi,
	\partial_s\bs\varphi,\partial_s\bs\psi,
	\partial_s^2\bs\varphi,\partial_s^2\bs\psi
	\bigr)(s).
	\notag
\end{align}
For a fixed configuration and state, let
\begin{align}
	\pi_s[\bs r,\bs R]:
	(\bbR^3)^6\rightarrow\bbR
	\notag
\end{align}
be linear in the six vector arguments of the second jet. We assume only that the scalar density is integrable for the configurations and virtual fields under consideration, so that its restriction to every material segment $[a,b]\subseteq[0,L]$ defines a {virtual-work functional}
\begin{align}
	\vWpi[\bs r,\bs R]
	\bigl([a,b],(\bs\varphi,\bs\psi\times)\bigr)
	:=
	\int_a^b
	\pi_s[\bs r,\bs R]
	\bigl(j_s^2(\bs\varphi,\bs\psi)\bigr)\,ds.
	\label{eq:workfunctional}
\end{align}

Consider a change of frame
\begin{align}
	\bs r^*
	=
	\bs c+\bs Q(\bs r-\bs o), \quad 
	\bs R^*=\bs Q\bs R, \quad 
	\bs\varphi^*=\bs Q\bs\varphi, \quad 
	\bs\psi^*=\bs Q\bs\psi,
\end{align}
where $\bs Q\in\mathrm{Orth}^+$ and $\bs c\in\bbR^3$ are independent of $s$. We say that the density is \emph{pointwise objective} if
\begin{align}
	\pi_s[\bs r^*,\bs R^*]
	\bigl(j_s^2(\bs\varphi^*,\bs\psi^*)\bigr)
	=
	\pi_s[\bs r,\bs R]
	\bigl(j_s^2(\bs\varphi,\bs\psi)\bigr)
	\label{eq:density-objectivity}
\end{align}
at every material point and for every second virtual jet.

Fix an origin $\bs o\in\bbR^3$, an infinitesimal translation vector $\bs c_0\in\bbR^3$, and an axial vector $\bs d\in\bbR^3$. For $\eps$ near zero, consider the family of superposed rigid motions
\begin{align}
	\bs r_\eps
	=
	\bs o+\eps\bs c_0
	+\exp(\eps\bs d\times)(\bs r-\bs o), \quad 
	\bs R_\eps
	=
	\exp(\eps\bs d\times)\bs R,
\end{align}
where $\bs d\times$ denotes the skew-symmetric tensor with axial vector $\bs d$. This family reduces to $(\bs r, \bs R)$ at $\eps=0$. Differentiating with respect to $\eps$ at $\eps=0$ gives
\begin{align}
	\bs\varphi_{\mathrm{rig}}
	:=
	\left.\frac{d}{d\eps}\bs r_\eps\right|_{\eps=0}
	=
	\bs c_0+\bs d\times(\bs r-\bs o),
	\quad
	\bs\psi_{\mathrm{rig}}
	=
	\bs d.
\end{align}
{Thus, $\bs\varphi_{\mathrm{rig}}$ and $\bs\psi_{\mathrm{rig}}$ are respectively the translational and axial rotational components of the associated infinitesimal rigid virtual displacement.}

We say that the density \emph{{vanishes for infinitesimal rigid virtual displacements}} if
\begin{align}
	\pi_s[\bs r,\bs R]
	\bigl(
	j_s^2(\bs\varphi_{\mathrm{rig}},\bs\psi_{\mathrm{rig}})
	\bigr)
	=0
	\label{eq:density-rigid-insensitivity}
\end{align}
{at every material point and for every infinitesimal rigid virtual displacement. Because $\pi_s$ is linear, this is equivalent to invariance under adding a second jet corresponding to an infinitesimal rigid virtual displacement. The pointwise requirements \eqref{eq:density-objectivity} and \eqref{eq:density-rigid-insensitivity} immediately give the corresponding properties of $\vWpi$ on every material segment.}

The following result is an algebraic characterization at a fixed material point. It identifies the most general local density that is linear in the second virtual jet, pointwise objective, and rigid-motion insensitive.

\begin{prop}\label{prop:local-work-representation}
{A linear second-jet density $\pi_s$ is pointwise objective and vanishes for infinitesimal rigid virtual displacements if and only if there exist unique material vectors $\msm,\msn,\msg,\msh\in\bbR^3$, whose dependence on the configuration and material point is suppressed, such that}
\begin{align}
	\msm^*=\msm,
	\quad
	\msn^*=\msn,
	\quad 
	\msg^*=\msg,
\quad 
	\msh^*=\msh,
	\notag
\end{align}
and
\begin{align}
	\pi_s[\bs r,\bs R]
	\bigl(j_s^2(\bs\varphi,\bs\psi)\bigr)
	=
	\msm\cdot\dmsu
	+\msn\cdot\dmsv
	+\msg\cdot\partial_s\dmsu
	+\msh\cdot\partial_s\dmsv.
	\label{eq:objective-linear-work-representation}
\end{align}
\end{prop}

\begin{proof}
Fix a material point $s$ and a configuration, and suppress their appearance in the notation. Let $\mathcal J_s^2\cong(\bbR^3)^6$ denote the second-jet space at $s$, and define the linear map
\begin{align}
	T_s:\mathcal J_s^2\rightarrow(\bbR^3)^4,
	\quad 
	T_s\bigl(j_s^2(\bs\varphi,\bs\psi)\bigr)
	:=
	\bigl(
	\dmsu,
	\dmsv,
	\partial_s\dmsu,
	\partial_s\dmsv
	\bigr).
	\notag
\end{align}
In terms of the second jet,
\begin{gather}
	\dmsu
	=\bs R^T\partial_s\bs\psi,
\quad 
	\dmsv
	=\bs R^T
	\bigl(\partial_s\bs\varphi+\bs v\times\bs\psi\bigr),
	\notag\\
	\partial_s\dmsu
	=\bs R^T
	\bigl(
	\partial_s^2\bs\psi
	-\bs u\times\partial_s\bs\psi
	\bigr),
	\quad 
	\partial_s\dmsv
	=\bs R^T\bigl[
	\partial_s^2\bs\varphi
	+\partial_s\bs v\times\bs\psi
	+\bs v\times\partial_s\bs\psi
	-\bs u\times
	\bigl(
	\partial_s\bs\varphi+\bs v\times\bs\psi
	\bigr)
	\bigr].
	\label{eq:virtual-strain-jet-map}
\end{gather}
The map $T_s$ is onto. Indeed, for arbitrary material vectors
$\bs\alpha,\bs\beta,\bs\gamma,\bs\delta$, prescribe a second jet at $s$ by
\begin{gather}
	\bs\varphi=\bs0,
	\quad 
	\bs\psi=\bs0,
	\notag\\
	\partial_s\bs\psi=\bs R\bs\alpha,
	\quad 
	\partial_s\bs\varphi=\bs R\bs\beta,
	\notag\\
	\partial_s^2\bs\psi
	=\bs R\bs\gamma+\bs u\times\bs R\bs\alpha,
	\quad 
	\partial_s^2\bs\varphi
	=\bs R\bs\delta
	-\bs v\times\bs R\bs\alpha
	+\bs u\times\bs R\bs\beta.
	\notag
\end{gather}
Substitution gives
\begin{align}
	T_s\bigl(j_s^2(\bs\varphi,\bs\psi)\bigr)
	=
	(\bs\alpha,\bs\beta,\bs\gamma,\bs\delta),
	\notag
\end{align}
which proves surjectivity.

Suppose that a difference jet
$j_s^2(\Delta\bs\varphi,\Delta\bs\psi)$ belongs to $\ker T_s$. Equation \eqref{eq:virtual-strain-jet-map} then gives
\begin{gather}
	\partial_s\Delta\bs\psi
	=\bs0,
	\quad 
	\partial_s\Delta\bs\varphi
	=\Delta\bs\psi\times\bs v,
	\notag\\
	\partial_s^2\Delta\bs\psi
	=\bs0,
	\quad 
	\partial_s^2\Delta\bs\varphi
	=\Delta\bs\psi\times\partial_s\bs v.
	\notag
\end{gather}
Set
\begin{align}
	\bs d
	:=\Delta\bs\psi,
	\quad
	\bs c_0
	:=
	\Delta\bs\varphi
	-\bs d\times(\bs r-\bs o).
	\notag
\end{align}
The value and first two derivatives of
\begin{align}
	\bs\varphi_{\mathrm{rig}}
	=
	\bs c_0+\bs d\times(\bs r-\bs o),
	\quad 
	\bs\psi_{\mathrm{rig}}
	=\bs d
	\notag
\end{align}
agree with the difference jet at $s$. Conversely, every rigid second jet is mapped to zero by $T_s$. Therefore,
\begin{align}
	\ker T_s
	=
	\left\{
	j_s^2(\bs\varphi_{\mathrm{rig}},\bs\psi_{\mathrm{rig}})
	\mid \bs c_0,\bs d\in\bbR^3
	\right\}.
	\notag
\end{align}

{Insensitivity to infinitesimal rigid virtual displacements} implies $\ker T_s\subseteq\ker\pi_s$; hence any two second jets having the same image under $T_s$ have the same value under $\pi_s$, and we may define the unique linear functional $\widehat\pi_s:(\bbR^3)^4\to\bbR$ by
\[
\widehat\pi_s(\bs\alpha,\bs\beta,\bs\gamma,\bs\delta):=\pi_s(J),
\]
where $J$ is any second jet satisfying $T_sJ=(\bs\alpha,\bs\beta,\bs\gamma,\bs\delta)$.
Every linear functional on $(\bbR^3)^4$ has a unique representation
\begin{align}
	\widehat\pi_s
	(\bs\alpha,\bs\beta,\bs\gamma,\bs\delta)
	=
	\msm\cdot\bs\alpha
	+\msn\cdot\bs\beta
	+\msg\cdot\bs\gamma
	+\msh\cdot\bs\delta.
	\notag
\end{align}
{This gives \eqref{eq:objective-linear-work-representation}. The four arguments of $\widehat\pi_s$ are variations of material strain vectors and are unchanged by a change of observer.} Pointwise objectivity and uniqueness therefore imply
$\msm^*=\msm$, $\msn^*=\msn$, $\msg^*=\msg$, and $\msh^*=\msh$.

Conversely, if the coefficient vectors are observer-invariant, the representation \eqref{eq:objective-linear-work-representation} is pointwise objective. All four variations of material strain vectors vanish on a rigid second jet, so the density also vanishes for infinitesimal rigid virtual displacements.
\end{proof}

It follows from Proposition~\ref{prop:local-work-representation} that, within the class of internal virtual-work functionals derived from densities that are linear in the second virtual jet, every pointwise objective density that vanishes on infinitesimal rigid virtual displacements generates an internal virtual-work functional on $[0,L]$ of the form
\begin{align}
	\vWint\bigl([0,L],(\bs\varphi,\bs\psi\times)\bigr)
	:=
	\int_0^L
	\left(
	\msm\cdot\dmsu
	+\msn\cdot\dmsv
	+\msg\cdot\partial_s\dmsu
	+\msh\cdot\partial_s\dmsv
	\right)\,ds.
	\label{eq:wint}
\end{align}
From this point onward, we restrict attention to coefficient fields satisfying
\begin{align}
	\msm,\msn,\msg,\msh
	\in L^{p'}([0,L];\mathbb R^3),
	\quad
	\frac{1}{p}+\frac{1}{p'}=1.
	\notag
\end{align}
{Their integrability is an additional analytic assumption on the class of internal virtual works studied below.} For each fixed admissible configuration, this assumption ensures that \eqref{eq:wint} defines a continuous linear functional on
$W^{2,p}([0,L];\bbR^3)\times W^{2,p}([0,L];\bbR^3)$, {the space of vector pairs $(\bs\varphi,\bs\psi)$ parameterizing the virtual displacements.}

{The fields $\msm$ and $\msn$ are work-conjugate to the virtual bending--twist and shear--extension strains, respectively, while $\msg$ and $\msh$ are work-conjugate to the spatial derivatives of those variations.} Their spatial representations are
\begin{align}
	\bs m:=\bs R\msm, \quad 
	\bs n:=\bs R\msn, \quad 
	\bs g:=\bs R\msg, \quad 
	\bs h:=\bs R\msh.
	\label{eq:spatialkinetics}
\end{align}
{The internal virtual work for any material segment $[a,b]\subseteq[0,L]$ is defined by}
\begin{align}
	\vWint\bigl([a,b],(\bs\varphi,\bs\psi\times)\bigr)
	:=
	\int_a^b
	\left(
	\msm\cdot\dmsu
	+\msn\cdot\dmsv
	+\msg\cdot\partial_s\dmsu
	+\msh\cdot\partial_s\dmsv
	\right)\,ds.
	\label{eq:wintloc}
\end{align}
{The external virtual work on the same segment is taken to be}
\begin{align}
	\vWext\bigl([a,b],(\bs\varphi,\bs\psi\times)\bigr)
	&:=
	\int_a^b
	\left(
	\tilde{\bs\ell}\cdot\bs\psi
	+\tilde{\bs f}\cdot\bs\varphi
	\right)\,ds
	+\bs m_b\cdot\bs\psi(b)
	-\bs m_a\cdot\bs\psi(a)
	\notag\\
	&\quad
	+\bs M_b\cdot\partial_s\bs\psi(b)
	-\bs M_a\cdot\partial_s\bs\psi(a)
	+\bs n_b\cdot\bs\varphi(b)
	-\bs n_a\cdot\bs\varphi(a)
	\notag\\
	&\quad
	+\bs N_b\cdot\partial_s\bs\varphi(b)
	-\bs N_a\cdot\partial_s\bs\varphi(a).
\end{align}
Here $\tilde{\bs\ell},\tilde{\bs f}\in L^{p'}([0,L];\mathbb R^3)$ are distributed load resultants, and
$\bs m_i,\bs M_i,\bs n_i,\bs N_i\in\mathbb R^3$, for $i=a,b$, are endpoint contact vectors. {The vectors $\bs m_i$ and $\bs n_i$ are work-conjugate to the endpoint values of $\bs\psi$ and $\bs\varphi$, whereas $\bs M_i$ and $\bs N_i$ are work-conjugate to their first derivatives with respect to $s$.} The displayed signs encode the outward orientations of the two ends of the segment.

At the physical ends $i=0,L$, these contact vectors represent prescribed end interactions, subject to the adopted orientation convention. At an internal cut, however, they are not independently prescribed data. {Section~3.2 derives them from the internal-work density and shows how the interactions on the two sides of a cut are related.}

{We postulate the \textbf{principle of virtual work} on the whole rod: at every time,}
\begin{align}
	\vWint\bigl([0,L],(\bs\varphi,\bs\psi\times)\bigr)
	=
	\vWext\bigl([0,L],(\bs\varphi,\bs\psi\times)\bigr)
	\label{eq:pvw}
\end{align}
{for every virtual displacement $(\bs\varphi,\bs\psi\times)$. No balance identity is assumed independently on proper subsegments of $[0,L]$. Section~3.2 shows that the whole-rod identity \eqref{eq:pvw}, together with the local representation of internal virtual work, determines the bulk balance equations and the induced contact interactions on internal cuts. These induced interactions yield a compatible balance of virtual work on every subsegment.}

\subsection{Local field equations and boundary interactions}
We first establish the regularity needed for integration by parts from the global weak identity, with the $L^{p'}$ coefficients and loads specified in Section 3.1. For a spatial vector field $\bs\chi$, define its director-covariant derivative by
\begin{align}
\mathscr D_s\bs\chi:=\partial_s\bs\chi-\bs u\times\bs\chi,
\quad
\partial_s(\bs R^T\bs\chi)=\bs R^T\mathscr D_s\bs\chi.
\end{align}

The derivatives are initially understood in the distributional sense. Introduce the force and couple resultants
\begin{align}
\bs F:=\bs n-\mathscr D_s\bs h,\quad
\bs C:=\bs m-\mathscr D_s\bs g.
\label{eq:force-couple-resultants}
\end{align}
{The work-conjugate fields $\bs n,\bs m$ need not themselves have endpoint values. The traces required by the contact virtual works are those of $\bs F,\bs C,\bs h,\bs g$.} The resultants $\bs F$ and $\bs C$ are the force and couple entering the momentum balances of Section 3.3.

\begin{prop}\label{lem:contact-regularity}
{Suppose that \eqref{eq:pvw} holds for every smooth virtual displacement,} with $(\bs r,\bs R)\in W^{2,p}$ and
$\bs m,\allowbreak\bs n,\allowbreak\bs g,\allowbreak\bs h,
\allowbreak\tilde{\bs f},\allowbreak\tilde{\bs\ell}\in L^{p'}([0,L];\bbR^3)$.
Then the distributions $\bs F,\bs C$ defined in \eqref{eq:force-couple-resultants} and the fields $\bs h,\bs g$ belong to $W^{1,p'}(I;\bbR^3)$. In particular, they have continuous representatives on $[0,L]$.
\end{prop}
\begin{proof}
The configuration regularity implies that $\bs u,\bs v$ are bounded and belong to $W^{1,p}$. Testing \eqref{eq:pvw} with $(\bs\varphi,\bs0)$, where $\bs\varphi\in C_c^\infty((0,L);\bbR^3)$, gives
\[
\int_0^L\left[
\bs h\cdot\partial_s^2\bs\varphi+
(\bs n+\bs u\times\bs h)\cdot\partial_s\bs\varphi
\right]\,ds
=\int_0^L\tilde{\bs f}\cdot\bs\varphi\,ds.
\]
Thus $\partial_s\bs F=-\tilde{\bs f}$ distributionally. A distribution on an interval whose derivative belongs to $L^{p'}$ is an absolutely continuous primitive of that derivative plus a constant, so $\bs F\in W^{1,p'}$. Moreover,
\[
\partial_s\bs h=\bs n+\bs u\times\bs h-\bs F\in L^{p'},
\]
which proves $\bs h\in W^{1,p'}$.

Next take $\bs\varphi=\bs0$ and $\bs\psi\in C_c^\infty((0,L);\bbR^3)$ to parameterize the {virtual displacement}. The regularity just obtained permits integration by parts in the terms involving $\bs h$ and reduces this identity to
\[
\int_I\left[
\bs g\cdot\partial_s^2\bs\psi+
(\bs m+\bs u\times\bs g)\cdot\partial_s\bs\psi
\right]\,ds
=\int_I(\tilde{\bs\ell}+\bs v\times\bs F)\cdot\bs\psi\,ds.
\]
Consequently, $\partial_s\bs C=-\tilde{\bs\ell}-\bs v\times\bs F\in L^{p'}$, so $\bs C\in W^{1,p'}$. Finally,
\[
\partial_s\bs g=\bs m+\bs u\times\bs g-\bs C\in L^{p'},
\]
and hence $\bs g\in W^{1,p'}$.
\end{proof}

With the regularity established in Proposition~\ref{lem:contact-regularity}, using \eqref{eq:dotu},\eqref{eq:dotv} and integrating by parts gives, on every segment $[a,b]$ and for every $W^{2,p}$ {virtual displacement},
\begin{align}
\vWint([a,b],(\bs\varphi,\bs\psi\times))
&=-\int_a^b\Bigl[\partial_s(\bs m-\mathscr D_s\bs g)
+\bs v\times(\bs n-\mathscr D_s\bs h)\Bigr]\cdot\bs\psi\,ds-\int_a^b\partial_s(\bs n-\mathscr D_s\bs h)\cdot\bs\varphi\,ds \notag\\
&\quad+\left[
\bs g\cdot\partial_s\bs\psi
+\bs h\cdot\partial_s\bs\varphi
+(\bs m-\mathscr D_s\bs g-\bs v\times\bs h)\cdot\bs\psi
+(\bs n-\mathscr D_s\bs h)\cdot\bs\varphi
\right]_a^b.
\label{eq:wint-localized}
\end{align}
For fields with this regularity, \eqref{eq:pvw} holds if and only if the bulk equations
\begin{align}
\begin{split}
\partial_s(\bs m-\mathscr D_s\bs g)
+\bs v\times(\bs n-\mathscr D_s\bs h)+\tilde{\bs\ell}&=\bs0,\\
\partial_s(\bs n-\mathscr D_s\bs h)+\tilde{\bs f}&=\bs0
\end{split}\label{eq:fieldequations0}
\end{align}
hold almost everywhere in $[0,L]$, and the prescribed physical-end coefficients equal the traces of the indicated continuous resultants, for $i=0,L$:
\begin{gather}
\bs M_i=\left.\bs g\right|_{s=i},
\quad
\bs N_i=\left.\bs h\right|_{s=i},
\label{eq:contactdoubleforce}\\
\bs m_i=\left.(\bs m-\mathscr D_s\bs g-\bs v\times\bs h)\right|_{s=i},
\label{eq:contactmoments}\\
\bs n_i=\left.(\bs n-\mathscr D_s\bs h)\right|_{s=i}.
\label{eq:contactforces}
\end{gather}
At each internal cut, define the contact coefficients by the same formulas \eqref{eq:contactdoubleforce}, \eqref{eq:contactmoments}, \eqref{eq:contactforces}, evaluated using the continuous representatives above. These coefficients are uniquely determined by the specified internal virtual-work density and the fixed outward-orientation convention. Substituting them and the bulk equations into \eqref{eq:wint-localized} gives the \textbf{local balance of virtual work}
\begin{align}
	\vWint([a,b], (\bs \varphi, \bs \psi \times)) = \vWext([a,b], (\bs \varphi, \bs \psi \times)). \label{eq:pvw_local}
\end{align}
for every $[a,b]\subseteq[0,L]$ and every virtual displacement in $W^{2,p}([a,b];\bbR^3\times\mathrm{Skw})$. Thus \eqref{eq:pvw_local} is a consequence of the whole-rod principle; no separate subinterval balance axiom has been imposed.

For simplicity, we adopt the classical inertial terms of special Cosserat rod theory. This choice is motivated by models of pantographic chains whose constituent bars undergo small deflections \cite{giorgio2017dynamics}. More general inertial terms may be derived within Germain's virtual-power framework \cite{Germain1973} by allowing the kinetic-energy density to depend additionally on material-coordinate derivatives of the centerline velocity and the cross-sectional angular velocity. We leave this extension for future work. With the classical inertial terms, the dynamical equations follow by decomposing the distributed powers into applied and inertial parts, as in the standard theory of special Cosserat rods (see Chapter 8 of \cite{antman05}):
\begin{align}
\begin{split}
	\tilde{\bs f} &= -\rho A \p_t^2 \bs r - \rho I_\alpha \p_t^2 \bs d_\alpha + \bs f, \\
	\tilde{\bs \ell} &= -\rho I_\alpha \bs d_\alpha \times \p_t^2 \bs r - \partial_t (\rho \bs J \bs w) + \bs \ell. 
\end{split}\label{eq:tildefell}
\end{align}
Here $\rho A(s)$ is the mass per unit reference length, and $\rho I_\alpha(s)$ are the generalized first moments of mass of the cross section. If $\bs e_1$ and $\bs e_2$ are aligned with the principal axes of inertia of the undeformed cross section, then
\begin{align}
\rho \bs J = I_{22} \bs d_1 \otimes \bs d_1 + I_{11} \bs d_2 \otimes \bs d_2 + (I_{11} + I_{22}) \bs d_3 \otimes \bs d_3,
\end{align}
where $I_{\alpha\alpha}(s)$ are the prescribed second mass moments of the cross section.

Substituting \eqref{eq:tildefell} into \eqref{eq:fieldequations0} gives the governing equations on $[0,L]$:
\begin{align}
\begin{split}
\partial_s(\bs m-\mathscr D_s\bs g)
+\bs v\times(\bs n-\mathscr D_s\bs h)+\bs\ell
&=\rho I_\alpha\bs d_\alpha\times\partial_t^2\bs r
+\partial_t(\rho\bs J\bs w),\\
\partial_s(\bs n-\mathscr D_s\bs h)+\bs f
&=\rho A\partial_t^2\bs r+\rho I_\alpha\partial_t^2\bs d_\alpha.
\end{split}\label{eq:fieldequations}
\end{align}
In the static hyperelastic case, with a stored energy obtained by three-dimensional reduction, these equations and boundary conditions reduce to equations (24)--(27) of \cite{Gupta2024}. This agreement provides an independent check on our direct intrinsic derivation.

\subsection{Balance laws}
Following Piola (see \cite{dellIsola2014b,Truesdell1960,Truesdell1965,RodriguezdellIsola2026}), we recover the classical momentum balances by testing the localized virtual-work identity with rigid virtual displacements.

We recall that an infinitesimal translation has virtual displacement $(\bs c_0,\bs0)$, with constant $\bs c_0\in\bbR^3$. Hence \eqref{eq:pvw_local} gives
\begin{align}
\vWext([a,b], (\bs c_0, \bs 0)) = 0. \label{eq:trans}
\end{align}
Because $\bs c_0$ is arbitrary, \eqref{eq:trans} is equivalent to
\begin{align}
\int_a^b\tilde{\bs f}\,ds
+\left[\bs F\right]_a^b=\bs0.
\end{align}
Using \eqref{eq:tildefell}, we obtain the \emph{balance of linear momentum}
\begin{align}
\frac{d}{dt}\int_a^b\bigl(\rho A\partial_t\bs r+\rho I_\alpha\partial_t\bs d_\alpha\bigr)\,ds
=\int_a^b\bs f\,ds
+\left[\bs F\right]_a^b.
\label{eq:balancelinmom}
\end{align}
Thus $\bs f$ is the applied body force per unit reference length, while $[\bs F]_a^b$ is the net contact force exerted on $[a,b]$ by the remainder of the rod.

Next, we recall an infinitesimal rotation
has virtual displacement $(\bs d\times(\bs r-\bs o),\bs d\times)$ for a fixed axial vector $\bs d \in \mathbb R^3$. Therefore
\begin{align}
\vWext([a,b],(\bs d\times(\bs r-\bs o),\bs d\times))=0.
\label{eq:rots}
\end{align}
Expanding the endpoint terms in \eqref{eq:rots} and using the arbitrariness of $\bs d$ gives
\begin{align}
\int_a^b\bigl(\tilde{\bs\ell}+(\bs r-\bs o)\times\tilde{\bs f}\bigr)\,ds
&+\bs m_b-\bs m_a
+(\bs r(b)-\bs o)\times\bs n_b-(\bs r(a)-\bs o)\times\bs n_a\notag\\
&+\partial_s\bs r(b)\times\bs N_b-\partial_s\bs r(a)\times\bs N_a=\bs0.
\end{align}
Substitution of \eqref{eq:contactdoubleforce}, \eqref{eq:contactmoments}, \eqref{eq:contactforces} makes the two terms involving $\bs v\times\bs h$ cancel. The result is the compact segment balance
\begin{align}
\int_a^b\bigl(\tilde{\bs\ell}+(\bs r-\bs o)\times\tilde{\bs f}\bigr)\,ds
+\left[(\bs r-\bs o)\times\bs F+\bs C\right]_a^b=\bs0.
\label{eq:balanceangmom0}
\end{align}
Finally, inserting the inertial decomposition \eqref{eq:tildefell} and collecting the time derivatives yields the \emph{balance of angular momentum}:
\begin{align}
\frac{d}{dt}\int_a^b\Bigl[
(\bs r-\bs o)\times\bigl(\rho A\partial_t\bs r&+\rho I_\alpha\partial_t\bs d_\alpha\bigr)
+\rho I_\alpha\bs d_\alpha\times\partial_t\bs r
+\rho\bs J\bs w
\Bigr]ds \\
&=\int_a^b\bigl((\bs r-\bs o)\times\bs f+\bs\ell\bigr)\,ds
+\left[(\bs r-\bs o)\times\bs F+\bs C\right]_a^b.
\end{align}
Thus $\bs\ell$ is the applied body couple per unit reference length, while $[\bs C]_a^b$ is the net contact couple.

\subsection{The theory for internal constraints} 
We now formulate the theory for the inextensibility and unshearability constraint
\begin{align}
	\msv = \bs e_3. \label{eq:constraint}
\end{align}
The simple modifications needed for unshearability alone, $v_1=v_2=0$, including retaining the axial active virtual work, are left to the reader. As in the previous subsections, we impose no boundary conditions on the kinematic fields in this subsection. This point is important below: prescribed endpoint values would change the domain and possibly the range of the linearized constraint operator.

Fix $p\in[1,\infty)$ and define
\begin{gather}
	\Phi : W^{2,p}([0,L];\bbE^3\times\mathrm{Orth}^+)
	\rar W^{1,p}([0,L];\bbR^3),\\
	\Phi(\bs r,\bs R):=\bs R^T\partial_s\bs r-\bs e_3=\msv-\bs e_3.
\end{gather}
For a virtual displacement $(\bs\varphi,\bs\psi\times)$, formulas \eqref{eq:virtvel} and \eqref{eq:dotv} give the differential
\begin{align}
	D\Phi_{(\bs r,\bs R)}(\bs\varphi,\bs\psi\times)
	=\dmsv
	=\bs R^T(\partial_s\bs\varphi+\bs v\times\bs\psi).
	\label{eq:linearized_constraint}
\end{align}
This derivative is surjective. Indeed, given any $\bs\eta\in W^{1,p}([0,L];\bbR^3)$, set $\bs\psi=\bs0$ and
\begin{align}
	\bs\varphi_{\bs\eta}(s):=\int_0^s\bs R(\xi)\bs\eta(\xi)\,d\xi.
\end{align}
Then $\bs\varphi_{\bs\eta}\in W^{2,p}([0,L];\bbR^3)$ and
$D\Phi_{(\bs r,\bs R)}(\bs\varphi_{\bs\eta},\bs0)=\bs\eta$. The map
$\bs\eta\mapsto(\bs\varphi_{\bs\eta},\bs0)$ is a bounded right inverse of $D\Phi_{(\bs r,\bs R)}$. Indeed, $\bs R \in W^{2,p}$ implies that $\bs R \in W^{1,\infty}$ and
\[
\|\bs\varphi_{\bs\eta}\|_{W^{2,p}}
\le C(L,\|\bs R\|_{W^{1,\infty}})\|\bs\eta\|_{W^{1,p}}.
\]
One-dimensional Sobolev product estimates also prove that $\Phi$ is smooth.
Thus $\Phi$ is a smooth submersion and
\begin{align}
	\cl M:=\Phi^{-1}(\{\bs0\})
	=\{(\bs r,\bs R)\in W^{2,p}([0,L];\bbE^3\times\mathrm{Orth}^+)\mid\msv=\bs e_3\}
\end{align}
is a smooth regular Banach submanifold. Its tangent space is
\begin{align}
	T_{(\bs r,\bs R)}\cl M
	=\left\{(\bs\varphi,\bs\psi\times)\in
	W^{2,p}([0,L];\bbR^3\times\mathrm{Skw})\mid\dmsv=\bs0\right\}.
\end{align}

Let $[a,b]\subseteq[0,L]$. As in Section 3.1, the active internal virtual work on $[a,b]$ is
\begin{align}
	\vWinta([a,b],(\bs\varphi,\bs\psi\times))
	:=\int_a^b\left(\msm\cdot\dmsu
	+\msg\cdot\partial_s\dmsu\right)\,ds,
	\label{eq:wintconstraint}
\end{align}
where $\msm,\msg\in L^{p'}([0,L];\bbR^3)$ are fixed with $\frac{1}{p} + \frac{1}{p'} = 1$. The \textbf{constrained principle of virtual work} requires
\begin{align}
	\vWinta([0,L],(\bs\varphi,\bs\psi\times))
	=\vWext([0,L],(\bs\varphi,\bs\psi\times))
	\label{eq:pvw2}
\end{align}
for every $(\bs\varphi,\bs\psi\times)$ satisfying $\dmsv=\bs0$ on $[0,L]$.

We use the following elementary multiplier result.
\begin{prop}\label{t:lagrangemultthm}
	Let $X$ and $Y$ be Banach spaces over $\bbK=\bbR$ or $\bbC$. Let $B\in X^*$ and let $A:X\rar Y$ be a bounded surjective linear operator. If $Bz=0$ for every $z\in\ker A$, then there exists a unique $\Lambda\in Y^*$ such that
	\begin{align}
		Bx=\Lambda(Ax),\quad x\in X.
	\end{align}
\end{prop}
\begin{proof}
	Since $A$ is surjective, the closed range theorem gives
	$\operatorname{Range}(A^*)=(\ker A)^\perp$. Hence $B=A^*\Lambda$ for some $\Lambda\in Y^*$. Moreover,
	$\ker A^*=(\operatorname{Range}A)^\perp=\{0\}$, so $\Lambda$ is unique.
\end{proof}

We set
\begin{gather}
	X:=W^{2,p}([0,L];\bbR^3\times\mathrm{Skw}), \quad
	Y:=W^{1,p}([0,L];\bbR^3),\\
	A(\bs\varphi,\bs\psi\times):=\dmsv, \quad
	B(\bs\varphi,\bs\psi\times):=\vWext([0,L],(\bs\varphi,\bs\psi\times))
	-\vWinta([0,L],(\bs\varphi,\bs\psi\times)).
\end{gather}
$B$ is continuous on $X$, and \eqref{eq:pvw2} implies that $B$ vanishes on $\ker A$. Thus, there exists a unique $\Lambda \in Y^*$ such that $B = \Lambda A$. 

By the representation of $(W^{1,p})^*$ (see Remark 21 of Section 8.3 of \cite{Brezis2011}), there exist global fields $\bs\lambda_1,\bs\lambda_2\in L^{p'}([0,L];\bbR^3)$ such that
\begin{align}
	B(\bs\varphi,\bs\psi\times)
	=\Lambda(\dmsv)
	=\int_0^L \left(\bs\lambda_1\cdot\dmsv
	+\bs\lambda_2\cdot\partial_s\dmsv\right)\,ds.
	\label{eq:multiplier_representation}
\end{align}
Relabeling
${\bs\lambda}_1=\msn$ and ${\bs\lambda}_2=\msh$, \eqref{eq:multiplier_representation} becomes
\begin{align}
	\vWinta([0,L],(\bs\varphi,\bs\psi\times))
	+\int_0^L\left(\msn\cdot\dmsv
	+\msh\cdot\partial_s\dmsv\right)\,ds=	\vWext([0,L],(\bs\varphi,\bs\psi\times)).
	\label{eq:constrained_full_work}
\end{align}

The identity \eqref{eq:constrained_full_work} holds for every virtual displacement in $X$, not only for tangent displacements satisfying $\dmsv=\bs0$, and is therefore formally equivalent to \eqref{eq:pvw}. Proposition~\ref{lem:contact-regularity} supplies the regularity of $\bs F,\bs C,\bs h,\bs g$ for any representing pair obtained above. Consequently, the bulk equations \eqref{eq:fieldequations0} hold almost everywhere, and \eqref{eq:fieldequations} follows with the inertial decomposition used there. The traces in \eqref{eq:contactdoubleforce}, \eqref{eq:contactmoments}, \eqref{eq:contactforces} reproduce the prescribed physical-end interactions and define induced internal-cut interactions for which the local identity \eqref{eq:pvw_local} holds. Since both $\dmsv$ and $\partial_s\dmsv$ vanish on admissible virtual displacements, the fields $\msn$ and $\msh$ are reactive: their virtual work vanishes on such displacements, and \textit{they are not prescribed by constitutive relations} on the constraint manifold.

We now discuss a \textit{striking difference} between the strain-gradient theory and the classical theory of Cosserat rods: for fixed active virtual work and external data, $\Lambda$ is unique, but the pair $(\msn,\msh)$ representing it is not. For any
$\bs\zeta\in W^{1,p'}([0,L];\bbR^3)$ with $\bs\zeta(0)=\bs\zeta(L)=\bs0$, the transformation
\begin{align}
	\msn^\#=\msn+\partial_s\bs\zeta,\quad
	\msh^\#=\msh+\bs\zeta
	\label{eq:reaction-gauge}
\end{align}
leaves \eqref{eq:constrained_full_work} unchanged, because
\[
\int_0^L\left(\partial_s\bs\zeta\cdot\bs\eta+
\bs\zeta\cdot\partial_s\bs\eta\right)\,ds
=[\bs\zeta\cdot\bs\eta]_0^L=0
\quad\text{for every }\bs\eta\in W^{1,p}.
\]
Conversely, two $L^{p'}$ pairs represent the same functional on $W^{1,p}$ only if their difference has this form: compactly supported tests give
$\partial_s(\msh^\#-\msh)=\msn^\#-\msn$ distributionally, and then independently prescribable endpoint traces force $\msh^\#-\msh$ to vanish at $0,L$. This freedom of representation is \textit{not present in the classical Cosserat theory}.

On a proper subinterval, the change in the reactive virtual work is instead
\begin{align}
	\int_a^b\left[
	(\msn^\#-\msn)\cdot\dmsv
	+(\msh^\#-\msh)\cdot\partial_s\dmsv
	\right]\,ds
	=[\bs\zeta\cdot\dmsv]_a^b.
	\label{eq:reaction-gauge-segment}
\end{align}
For fixed active fields $\msm,\msg$, the spatial fields change by
$\bs n^\#-\bs n=\bs R\partial_s\bs\zeta$ and
$\bs h^\#-\bs h=\bs R\bs\zeta$. Since
$\mathscr D_s(\bs R\bs\zeta)=\bs R\partial_s\bs\zeta$, both resultants $\bs F,\bs C$ remain unchanged, but the induced contact coefficients (including contact double-forces and double-couples) at a cut $i$ change according to
\begin{gather}
	\bs N_i^\#-\bs N_i=\bs R(i)\bs\zeta(i), \quad 
	\bs m_i^\#-\bs m_i=-\bs v(i)\times\bs R(i)\bs\zeta(i),\notag\\
	\bs M_i^\#-\bs M_i=\bs0, \quad 
	\bs n_i^\#-\bs n_i=\bs0.
	\label{eq:reaction-gauge-contacts}
\end{gather}
Thus, each representative pair $(\msn, \msh)$ for $\Lambda$ yields a compatible local virtual-work identity with its own induced cut interactions. The physical-end interactions are unchanged, and all representatives give the same virtual work on constrained virtual displacements.

Accordingly, the global constrained principle uniquely determines the multiplier functional $\Lambda$. Equivalently, it determines the gauge-invariant bulk combination
\[
\msn-\partial_s\msh=\bs R^T\bs F
\]
together with the endpoint traces of $\msh$. It does not, however, determine the fields $\msn$ and $\msh$ separately, nor does it select a unique extension of the constrained virtual work to arbitrary virtual displacements on each internal segment. Selecting a particular representative pair $(\msn,\msh)$ therefore requires an additional modeling or normalization convention (one such convention will be suggested in Section 4.2).

\section{Constitutive theory}

This section develops constitutive relations that close the field equations derived in Section 3. We first show that objectivity reduces the constitutive dependence to the histories of the material strains and their spatial derivatives, and then define the corresponding elastic, hyperelastic, and viscoelastic specializations. For inextensible and unshearable rods, we distinguish the active constitutive response from the reactive force and double-force fields, whose representation possesses a gauge freedom, and use a heuristic penalty argument to select a natural three-parameter family of representatives. Finally, we formulate transverse hemitropy, transverse isotropy, and flip symmetry for strain-gradient rods and use invariant theory to represent symmetric stored-energy densities in terms of polynomial integrity bases.

\subsection{Objectivity}
We first discuss constitutive relations that close the field equations of the form: 
\begin{align}
	\bs m(s,t) &= \tilde{\bs m}\bigl (\bs r^t(s, \cdot), \p_s \bs r^t(s, \cdot), \p_s^2 \bs r^t(s, \cdot), \bs d_\alpha^t(s, \cdot), \p_s \bs d_\alpha^t(s, \cdot), \p_s^2 \bs d_\alpha^t(s, \cdot), s\bigr ), \label{eq:mconst}\\
	\bs n(s,t) &= \tilde{\bs n}\bigl (\bs r^t(s, \cdot), \p_s \bs r^t(s, \cdot), \p_s^2 \bs r^t(s, \cdot), \bs d_\alpha^t(s, \cdot), \p_s \bs d_\alpha^t(s, \cdot), \p_s^2 \bs d_\alpha^t(s, \cdot), s\bigr ), \label{eq:nconst}\\
	\bs g(s,t) &= \tilde{\bs g}\bigl (\bs r^t(s, \cdot), \p_s \bs r^t(s, \cdot), \p_s^2 \bs r^t(s, \cdot), \bs d_\alpha^t(s, \cdot), \p_s \bs d_\alpha^t(s, \cdot), \p_s^2 \bs d_\alpha^t(s, \cdot), s\bigr ), \label{eq:gconst} \\
	\bs h(s,t) &= \tilde{\bs h}\bigl (\bs r^t(s, \cdot), \p_s \bs r^t(s, \cdot), \p_s^2 \bs r^t(s, \cdot), \bs d_\alpha^t(s, \cdot), \p_s \bs d_\alpha^t(s, \cdot), \p_s^2 \bs d_\alpha^t(s, \cdot), s\bigr ). \label{eq:hconst} 
\end{align}
where $\bs m$, $\bs n$, $\bs g$, and $\bs h$ are defined in \eqref{eq:spatialkinetics}, and for a vector field $\bs a: [0,L] \times [0,T] \rar \bbR^3$, $\bs a^t(s, \cdot) := \{ \bs a(s,\tau) \mid \tau \in [0,t] \}$ is its past history.

We say that a vector-valued constitutive function
\begin{align}
	\tilde{\bs a}
	=
	\tilde{\bs a}
	\bigl(
	\bs r^t,
	\partial_s \bs r^t,
	\partial_s^2 \bs r^t,
	\bs d_\alpha^t,
	\partial_s \bs d_\alpha^t,
	\partial_s^2 \bs d_\alpha^t,
	s
	\bigr)
\end{align}
is \textit{objective} if, under every change of frame,
\begin{align}
	(\bs r^*(\cdot,t),\bs R^*(\cdot,t))
	=
	\bigl(\bs o + 
	\bs c(t)+\bs Q(t)(\bs r(\cdot,t)-\bs o),
	\bs Q(t)\bs R(\cdot,t)
	\bigr),
	\label{eq:change_of_frame2}
\end{align}
with $\bs Q(t)\in\mathrm{Orth}^+$ and $\bs c(t)\in\bbR^3$, it transforms according to
\begin{align}
\label{eq:transformation_law}
	\tilde{\bs a}
	\bigl(
	(\bs r^*)^t,
	\partial_s (\bs r^*)^t,
	\partial_s^2 (\bs r^*)^t,
	(\bs d^*)_\alpha^t,
	\partial_s (\bs d^*)_\alpha^t,
	\partial_s^2 (\bs d^*)_\alpha^t,
	s
	\bigr)
	=
	\bs Q\,
	\tilde{\bs a}
	\bigl(
	\bs r^t,
	\partial_s \bs r^t,
	\partial_s^2 \bs r^t,
	\bs d_\alpha^t,
	\partial_s \bs d_\alpha^t,
	\partial_s^2 \bs d_\alpha^t,
	s
	\bigr).
\end{align}
The following proof is modeled on that of \cite{antman05}, Chapter 8. 
\begin{prop}
The most general objective constitutive relations of the form \eqref{eq:mconst}-\eqref{eq:hconst} are 
\begin{align}
\label{eq:objective_constitutive_equations}
\msm(s,t) &= \hat{\msm}\bigl (\msu^t(s, \cdot), \msv^t(s, \cdot), \p_s \msu^t(s, \cdot), \p_s \msv^t(s, \cdot), s\bigr ), \\
\msn(s,t) &= \hat{\msn} \bigl (\msu^t(s, \cdot), \msv^t(s, \cdot), \p_s \msu^t(s, \cdot), \p_s \msv^t(s, \cdot), s\bigr ), \\
\msg(s,t) &= \hat{\msg} \bigl (\msu^t(s, \cdot), \msv^t(s, \cdot), \p_s \msu^t(s, \cdot), \p_s \msv^t(s, \cdot), s\bigr ), \\
\msh(s,t) &= \hat{\msh}\bigl (\msu^t(s, \cdot), \msv^t(s, \cdot), \p_s \msu^t(s, \cdot), \p_s \msv^t(s, \cdot), s\bigr ). 
\end{align}
\end{prop}
 \begin{proof}
   We prove the proposition for $\bs m$; the argument for the remaining variables is nearly identical. Assume that $\tilde{\bs m}$ satisfies \eqref{eq:transformation_law}. 
    Setting $\bs Q(t) = \bs I$ and $\bs c(t)$ arbitrary reveals that $\tilde {\bs m }$ does not depend on $\bs r^t$, which allows us to set
    \begin{align}
    \label{eq:objective_representation}
        \tilde{\bs m}(\bs r^t, \partial_s \bs r^t, \bs \partial_s^2 \bs r^t, \bs d_\alpha^t, \partial_s \bs d_\alpha^t, \partial_s^2 \bs d_\alpha^t) = \hat{ m}_k(\msu^t, \msv^t, \p_s \msu^t, \p_s \msv^t, \bs d^t_\alpha, s)\bs d_k, \quad \hat{m}_k = \tilde{\bs m} \cdot \bs d_k.
    \end{align}
    Substituting \eqref{eq:objective_representation} into \eqref{eq:transformation_law} yields 
    \begin{align}
        \hat{m}_k (\msu^t, \msv^t, \p_s \msu^t, \p_s \msv^t, (\bs Q \bs d_\alpha)^t, s) \bs Q \bs d_k = \hat{m}_k(\msu^t, \msv^t, \p_s \msu^t, \p_s \msv^t, \bs d_\alpha^t,s) \bs Q \bs d_k 
    \end{align}
    Fixing $s = s_0 \in [0,L]$ and setting $\bs Q(t) = \bs R^T(s_0,t)$ now implies that 
    \begin{align}
        \hat{m}_k(\msu^t, \msv^t, \p_s \msu^t, \p_s \msv^t, \bs e^t_\alpha, s_0) = \hat{m}_k(\msu^t, \msv^t, \p_s \msu^t, \p_s \msv^t, \bs d^t_\alpha, s_0), \quad s_0 \in [0, L]. 
    \end{align}
   It follows that $\hat{m}_k$ is independent of $\bs d^t_\alpha$, so that pulling back by $\bs R^T(s, t)$ yields
    \begin{align}
        \msm = \bs R^T {\bs m} = \hat{m}_k(\msu, \msv, \p_s \msu, \p_s \msv, s)\bs e_k =: \hat{\msm}(\msu^t, \msv^t, \p_s \msu^t, \p_s \msv^t, s). 
    \end{align}
\end{proof}

We define a \textit{strain-gradient elastic Cosserat rod} as one whose constitutive functions depend only on the quantities $(\msu, \msv, \p_s \msu, \p_s \msv, s)$. The rod is \emph{hyperelastic} if there exists a scalar-valued function $W = \hat{W}(\msu, \msv, \p_s \msu, \p_s \msv, s)$, called the \emph{stored energy density}, such that 
\begin{align}
    \label{eq:hyperelasticity}
    \hat{\msm} = \frac{\partial \hat W}{\p \msu}, \quad \hat{\msn} = \frac{\p \hat W}{\p \msv}, \quad
    \hat{\msg} = \frac{\p \hat W}{\p \msu_s}, \quad
    \hat{\msh} = \frac{\p \hat W }{\p \msv_s},
\end{align}
where $\bs a_s := \p_s \bs a$ going forward. A \emph{strain-gradient viscoelastic Cosserat rod of complexity $k$} is one whose constitutive functions depend on the $j$th-order time derivatives of 
$
	(\msu, \msv, \msu_s, \msv_s, s),
$
for $0 \leq j \leq k$.

\subsection{Unshearable and inextensible rods}
We now discuss the constitutive theory for unshearable and inextensible rods, leaving the minor modifications for solely unshearable rods to the reader.
In the case of the internal unshearable and inextensible constraint, 
\begin{align}
	\msv = \bs e_3, 
\end{align}
the objective constitutive relations prescribed take the form
\begin{align}
\label{eq:objective_constitutive_equations_inext}
\msm(s,t) = \hat{\msm}\bigl (\msu^t(s, \cdot), \p_s \msu^t(s, \cdot), s\bigr ), \quad
\msg(s,t) = \hat{\msg} \bigl (\msu^t(s, \cdot), \p_s \msu^t(s, \cdot), s\bigr ).
\end{align}
We define an \textit{inextensible strain-gradient elastic Cosserat rod} as one whose constitutive functions depend only on the quantities $(\msu, \p_s \msu)$. The rod is \emph{hyperelastic} if there exists a scalar-valued stored-energy function $W = \hat{W}(\msu, \p_s \msu, s)$ such that 
\begin{align}
	\label{eq:hyperelasticity_inext}
	\hat{\msm} = \frac{\partial \hat W}{\p \msu}, \quad
	\hat{\msg} = \frac{\p \hat W}{\p \msu_s}.
\end{align}

As discussed in Section 3.4, in contrast to the classical Cosserat rod theory there is a freedom of representation for the reactive virtual work: for any
$\bs\zeta\in W^{1,p'}([0,L];\bbR^3)$ with $\bs\zeta(0)=\bs\zeta(L)=\bs0$, the transformation
\begin{align}
	\msn^\#=\msn+\partial_s\bs\zeta,\quad
	\msh^\#=\msh+\bs\zeta
	\label{eq:reaction-gauge2}
\end{align}
leaves \eqref{eq:constrained_full_work} unchanged. Accordingly, only the gauge-invariant bulk combination
\[
\msn-\partial_s\msh
\]
together with the endpoint traces of $\msh$ are determined from \eqref{eq:constrained_full_work}. It does not, however, determine the fields $\msn$ and $\msh$ separately. Thus, a relation between $\msn$ and $\msh$ must be postulated to eliminate the gauge freedom.

One natural choice is to postulate that the rod does not support double forces, and therefore,
\begin{align}
	\msh = \bs 0. \label{eq:hzero}
\end{align}
Then, for an inextensible hyperelastic strain-gradient rod with stored energy $W(\msu, \msu_s, s)$ and prescribed external forces, \eqref{eq:constrained_full_work} takes the familiar stationary form
\begin{align}
	\frac{d}{d\eps}\Big|_{\eps=0}
	\int_0^L\left[
	\hat W(\msu_\eps,\partial_s\msu_\eps,s)
	+\msn\cdot(\msv_\eps-\bs e_3)
	\right]ds
	=\vWext([0,L],(\bs\varphi,\bs\psi\times)),
\end{align}
for every {virtual displacement} in $W^{2,p}$ with $\msn$ playing the role of a Lagrange multiplier. Moreover, in the case that the stored energy is independent of $\msu_s$, the theory reduces to the classical Cosserat rod theory.

The condition \eqref{eq:hzero} is the zero-length-scale member of a natural three-parameter family of gauge-fixing relations suggested by a penalty approximation. The following argument is heuristic: we assume a smooth family of penalized equilibria and finite limiting reaction fields, without proving convergence as the penalty parameter tends to zero.

Let $a_i>0$ and $b_i\geq0$ be constants, and define the diagonal material tensors
\begin{align}
	\mathsf A
	:=\sum_{i=1}^3a_i\bs e_i\otimes\bs e_i,
	\quad
	\mathsf B
	:=\sum_{i=1}^3b_i\bs e_i\otimes\bs e_i,
	\quad
	\mathsf D
	:=\mathsf B\mathsf A^{-1}
	=\sum_{i=1}^3\ell_i^2\bs e_i\otimes\bs e_i,
	\quad \ell_i^2:=\frac{b_i}{a_i}.
	\label{eq:penalty-tensors}
\end{align}
Instead of imposing $\msv=\bs e_3$ exactly, consider a slightly extensible and shearable hyperelastic rod with stored energy
\begin{align}
	\hat W_\sigma
	(\msu,\partial_s\msu,\msv,\partial_s\msv)
	:=\hat W(\msu,\partial_s\msu)
	+\frac{1}{2\sigma}
	(\msv-\bs e_3)\cdot\mathsf A(\msv-\bs e_3)
	+\frac{1}{2\sigma}
	\partial_s\msv\cdot\mathsf B\partial_s\msv.
	\label{eq:penalized-energy}
\end{align}
The fields conjugate to $\msv$ and $\partial_s\msv$ are
\begin{align}
	\msn_\sigma
	=\frac{1}{\sigma}\mathsf A(\msv_\sigma-\bs e_3), \quad 
	\msh_\sigma
	=\frac{1}{\sigma}\mathsf B\partial_s\msv_\sigma.
	\label{eq:penalized-reactions}
\end{align}
Since $\mathsf A$ and $\mathsf B$ are constant, these identities imply the exact relation
\begin{align}
	\msh_\sigma
	=\mathsf D\partial_s\msn_\sigma.
	\label{eq:penalty-reaction-relation}
\end{align}
Provided that $\hat W$ is bounded below, a uniform bound on the penalized energy formally drives $\msv_\sigma$ to $\bs e_3$ as $\sigma\to0$ and, in every component for which $b_i>0$, drives $\partial_s\msv_\sigma$ to zero. We therefore assume, as part of this formal limiting argument, that a family of penalized equilibria can be chosen for which $\msn_\sigma$ and $\msh_\sigma$ approach finite limits $\msn$ and $\msh$. Passing formally to the limit in \eqref{eq:penalty-reaction-relation} gives
\begin{align}
	\msh
	=\mathsf D\partial_s\msn.
	\label{eq:limiting-reaction-relation}
\end{align}

This relation can be expressed directly in terms of the contact force and contact double force. By \eqref{eq:force-couple-resultants}, \eqref{eq:contactdoubleforce}, and \eqref{eq:contactforces},
\begin{align}
	\bs h
	=\bs R\mathsf D\partial_s\msn, \quad 
	\bs F
	=\bs R\left(\msn-\mathsf D\partial_s^2\msn\right),
	\label{eq:penalty-contact-relations}
\end{align}
and, at every physical end or internal cut $s=i$,
\begin{align}
	\bs N_i
	=\left.\bs R\mathsf D\partial_s\msn\right|_{s=i}, \quad 
	\bs n_i
	=\left.\bs R\left(\msn-\mathsf D\partial_s^2\msn\right)\right|_{s=i}.
	\label{eq:penalty-cut-relations}
\end{align}
If every $\ell_i>0$, eliminating $\msn$ from \eqref{eq:penalty-contact-relations} gives the equivalent relation
\begin{align}
	\partial_s(\bs R^T\bs F)
	=\mathsf D^{-1}\bs R^T\bs h
	-\partial_s^2(\bs R^T\bs h).
	\label{eq:force-double-force-relation}
\end{align}
When some $\ell_i$ vanishes, the corresponding component of \eqref{eq:limiting-reaction-relation} simply gives $h_i=0$.

Finally, we see that \eqref{eq:limiting-reaction-relation} fixes the gauge. Indeed, suppose that the gauge-related pair
$(\msn+\partial_s\bs\zeta,\msh+\bs\zeta)$ also satisfies \eqref{eq:limiting-reaction-relation}, where $\bs\zeta(0)=\bs\zeta(L)=\bs0$. Subtracting the two relations gives
\begin{align}
	\bs\zeta-\mathsf D\partial_s^2\bs\zeta=\bs0.
\end{align}
Taking the inner product with $\bs\zeta$, integrating over $[0,L]$, and using the endpoint conditions yields
\begin{align}
	\int_0^L
	\left(
	|\bs\zeta|^2
	+\partial_s\bs\zeta\cdot\mathsf D\partial_s\bs\zeta
	\right)\,ds
	=0.
\end{align}
Since $\mathsf D$ is positive semidefinite, $\bs\zeta=\bs0$. Hence the three nonnegative parameters $\ell_i^2=b_i/a_i$ select a unique representative for the reactive virtual work. The special choice $\ell_1=\ell_2=\ell_3=0$ gives $\msh=\bs0$ and recovers \eqref{eq:hzero}. This conclusion should be understood as the representative selected by the particular penalty regularization \eqref{eq:penalized-energy}, not as a consequence of the inextensibility and unshearability constraint alone.

\subsection{Transverse hemitropy and isotropy}
Our discussion of material symmetry for a strain-gradient Cosserat rod extends that of classical Cosserat rods, and for simplicity of exposition, we specialize \cite{Healy2002,antman05} to strain-gradient hyperelastic Cosserat rods. We assume that the reference configuration of the rod is initially straight and aligned with $\bs e_3$ and define 
\begin{align}
\msQ_\theta = \begin{pmatrix}
    \cos \theta & -\sin \theta & 0 \\
    \sin \theta & \cos \theta & 0 \\
    0 & 0 & 1 
\end{pmatrix}, 0\leq \theta < 2\pi, \quad  \msE = \begin{pmatrix}
    1 & 0 & 0 \\
    0 & -1 & 0 \\
    0 & 0 & 1 
\end{pmatrix}
\label{eq:rotationfamily}
\end{align}
with respect to the standard basis $\{\bs e_1, \bs e_2, \bs e_3\}$. We define the \emph{reference symmetry groups} 
\begin{align}
   \cl O_c(2) := \{\msQ_\theta, \msE \msQ_{\theta} \mid 0 \leq \theta < 2\pi\}, \quad \cl R_c(2) := \{ \msQ_\theta \mid 0 \leq \theta <2\pi \}. 
\end{align}
Consider the motions described by $(\bs r(s,t), \bs R(s,t))$ and \[(\bs r(s,t), \bs R^\#(s,t)):= (\bs r, \bs R(s,t) \msT)\] where $\msT \in \mathcal{R}_c(2)$. Letting $(\msu, \msv)$ and $(\msu^\#, \msv^\#)$ denote the corresponding strains, we compute
\begin{align}
    \msv^\# &= (\bs R^\#)^T \bs v = \msT^T \msv \label{eq:strain_relations_1}\\
    \msu^\# &= \text{ax}((\bs R^\#)^T \p_s \bs R^\# )= \text{ax}(\msT^T \bs R^T\partial_s \bs R \msT) = \det(\msT) \msT^T \msu. \\
    \p_s\msv^\# & = \msT^T \msv_s =: \msv_s^\#\\
    \p_s \msu^\# &= \det(\msT) \msT^T \msu_s =: \msu_s^\# \label{eq:strain_relations_4}
\end{align}
Of course $\det(\msT) = 1$ for a rotation, but we see from this computation that if we allow $\msT=\msE$, a reflection, which is not an admissible motion for Cosserat rods, the relations \eqref{eq:strain_relations_1}-\eqref{eq:strain_relations_4} are still valid. Allowing reflections considers a rod and its ``mirror image'' to be constitutively comparable when the mirror is parallel to the rod. We define a hyperelastic special Cosserat rod to be \emph{transverse hemitropic} if, for all $\msT \in \cl R_c(2)$, we have 
\begin{align}
\label{eq:transversehemitropy}
\hat{W}(\msu^\#, \msv^\#, \msu^\#_s, \msv^\#_s,s) = \hat{W}(\msu, \msv, \msu_s, \msv_s, s). 
\end{align}
The rod is said to be \emph{transverse isotropic} if $\hat{W}$ satisfies \eqref{eq:transversehemitropy} for all $\msT \in \cl O_c(2)$, i.e., we also have
\begin{align}
    \label{eq:transverseisotropy}
   \hat{W}(-\msE \msu, \msE \msv, -\msE \msu_s, \msE \msv_s, s) = \hat{W}(\msu, \msv, \msu_s, \msv_s,s).
\end{align} This distinction separates hemitropic rods that may be chiral, such as DNA and twine, from transversely isotropic rods with reflection symmetry, such as idealized homogeneous rubber rods.

A distinct symmetry condition to \eqref{eq:transverseisotropy} is that of \emph{flip-symmetry}, which requires that {the} stored energy of the rod cannot distinguish between the motions \[(\bs r(s, t), \bs R(s,t)),  \quad  (\bs r(2p-s, t), -\bs R(2p-s, t)\msE)\] for all admissible $p$, when evaluated at $s = p$ \cite{Healy2002}. Intuitively, flipping our rod $180^\circ$ around an axis in the $\bs e_2$ cross-sectional direction and reversing the parameterization of the rod does not change its material response. See the discussion in section 6 of \cite{Healy2002} for motivation and analysis. Applying this definition to our {strain-gradient} rod, we say a rod is flip-symmetric if $\hat{W}$ satisfies
\begin{align}
    \label{eq:flipsymmetry}
    \hat{W}(\msE \msu, \msE \msv, -\msE \msu_s, -\msE \msv_s, s) = \hat{W}(\msu, \msv, \msu_s, \msv_s,s).
\end{align}
The sign change on the higher gradient terms comes from the reverse parameterization of the rod.\footnote{It should be noted that \eqref{eq:transverseisotropy} and \eqref{eq:flipsymmetry} are relations on $\hat{W}$, and do not necessarily arise from proper motions. We assume that $\hat{W}$ can have a larger domain than the admissible domain for strains representing proper motions. The action yielding transverse hemitropy \eqref{eq:transversehemitropy}, on the other hand, does preserve proper motions.}

For any \(\msa \in \mathbb R^3\), define \(\msa_\perp := \msa -(\msa \cdot \bs e_3)\bs e_3\). An example of an energy that is transverse isotropic and flip-symmetric is 

\begin{align}
	\label{eq:stored_energy}
	W(s) = \hat{W}(\msu, \msu_s, s) = \frac{1}{2}A(s)|\msu_\perp(s)|^2 + \frac{1}{2}B(s)|(\msu_s)_\perp(s)|^2 + \frac{1}{2}\xi(s)u_3^2(s) +  \frac{1}{2}\zeta(s)u_{3,s}^2, 
\end{align}
for some scalar functions \(A, B, \xi, \zeta\). This energy is natural for inextensible and unshearable rods, and will be discussed further below. A similar type of energy was derived in \cite{Gupta2024} as a result of dimensional reduction. Here, rod-material symmetry (which does not have an obvious analog in three-dimensional theories) yields our constitutive representation. The direct approach will prove indispensable in analyzing exotic rod energies, where there are no obvious three-dimensional parent energies available.  

When restricted to planar, untwisted rod deformations, \eqref{eq:stored_energy} generalizes the energy of the inextensible strain-gradient rod studied in \cite{rizzi26}. Adding a term that changes sign under the simultaneous transformation
\[
(u_3,u_{3,s})\mapsto(-u_3,-u_{3,s}),
\]
such as a term linear in either \(u_3\) or \(u_{3,s}\), preserves transverse hemitropy but breaks transverse isotropy, thereby allowing a chiral constitutive response.

When $\hat{W}$ is invariant under \eqref{eq:strain_relations_1}-\eqref{eq:strain_relations_4}, differentiation gives the following equivariance relations
\begin{align}
    \label{eq:constitutivefunctionshemitropy}
    \hat{\msm}(\msu, \msv, \msu_s, \msv_s) &= \partial_{\msu}\left[\hat{W}(\msu^\#, \msv^\#, \msu_s^\#, \msv_s^\#) \right] = \det(\msT) \msT \hat{\msm}(\msu^\#, \msv^\#, \msu_s^\#, \msv_s^\#), \\
    \hat{\msn}(\msu, \msv, \msu_s, \msv_s) &= \msT \hat{\msn}(\msu^\#, \msv^\#, \msu_s^\#, \msv_s^\#)\\
    \hat{\msg}(\msu, \msv, \msu_s, \msv_s) &= \det(\msT) \msT \hat{\msg}(\msu^\#, \msv^\#, \msu_s^\#, \msv_s^\#)\\
    \hat{\msh}(\msu, \msv, \msu_s, \msv_s) &= \msT \hat{\msh}(\msu^\#, \msv^\#, \msu_s^\#, \msv_s^\#). 
\end{align}
These relations serve as the definition of transverse hemitropy or isotropy when no stored energy exists. The corresponding equivariance relations for flip symmetry follow from the same procedure.

\subsection{Invariant representations and integrity bases of stored energy densities} 
Symmetries of the stored energy density $\hat{W}$ correspond to alternate functional representations of $\hat{W}$ where the symmetries are encoded in the domain. The determination of these representations plays a fundamental role in rod theory and mechanics as a whole \cite{Healy2002, antman05, Smith1994}. We approach this problem via an abstraction of \cite{Healy2002, antman05} using the theory of Lie groups. This allows us to use general topological results in \cite{Weyl1946, Schwarz1975} on groups acting orthogonally on $\bbR^n$ and their effect on functions invariant under these groups. The Lie group framework thus proves to be the natural context in which to study these symmetries. For a general reference on Lie groups in this context, see \cite{Bredon1972}.

Let $\cl G$ be any compact Lie group that acts orthogonally on $\bbR^{n}$.  By the Hilbert-Weyl finite generation theorem \cite{Weyl1946}, there exists a finite set $\cl P := \{p_1, \ldots p_k\}$ of polynomials $p_i: \bbR^n \to \bbR$ that generate the algebra $\cl P[\bbR^n]^{\cl G}$ of invariant polynomials under $\cl G$. We define the basis map $\rho: \bbR^n \to \bbR^k$ as $\rho = (p_1, \ldots p_k)$. 
We say a function $f \in C^\infty(\bbR^{n})$ is invariant under $\cl G$ if $f(\cl T \msZ) = f(\msZ)$ for all $\cl T \in \cl G, \msZ \in \bbR^{n}$, and we write $f \in C^\infty(\bbR^n)^{\cl G}$. For such an $f$, \cite{Schwarz1975} guarantees the existence of a function $\Phi \in C^\infty(\bbR^k)$  such that $f= \Phi \circ \rho$.

Thus, for any such $\cl G$, if one can determine a (not necessarily unique) generating set $\{p_i\}$ of polynomials, one recovers a representation of $\hat{W}$ in terms of these invariants. Such a generating set is called an \emph{integrity basis} of $\cl G$. 
Our goal in this section is to determine useful representations of $\hat{W}$ in terms of appropriate integrity bases \cite{Smith1994, Schwarz1975}. To determine these representations systematically, we define $\msZ = (\msu, \msv, \msu_s, \msv_s) \in \bbR^{12}$ and the maps $\sigma_{\text{iso}}, \sigma_{\text{flip}}: \cl O_c(2) \to  O(12)$ by 
\begin{align}
    \sigma_{\text{iso}}(\msT)\msZ &= (\det(\msT) \msT \msu, \msT \msv, \det(\msT)\msT \msu_s, \msT \msv_s),\\
    \sigma_{\text{flip}}(\msT)\msZ &= (\msT \msu, \msT \msv, \det(\msT)\msT \msu_s, \det(\msT)\msT \msv_s).
\end{align}
We define the groups
\begin{gather}
    \cl G_0 := \sigma_{\text{iso}}(\cl R_c(2))  = \sigma_{\text{flip}}(\cl R_c(2))\\
    \cl G_1 := \sigma_{\text{iso}}(\cl O_c(2)), \quad \cl G_2 := \sigma_{\text{flip}}(\cl O_c(2)).
\end{gather}
 It follows that, for each $i = 0,1,2$, $\cl G_i$ is a compact Lie group that acts orthogonally on $\bbR^{12} $. 
We assume that for each $s \in [0,L]$, $\hat{W}(\cdot, s) \in C^\infty(\bbR^{12})$. Conditions \eqref{eq:transversehemitropy} (hemitropy), \eqref{eq:transverseisotropy} (isotropy), and \eqref{eq:transversehemitropy} and \eqref{eq:flipsymmetry} (hemitropy and flip-symmetry) can now be restated as, for each $s \in [0,L]$, \[\hat{W}(\msZ, s) = \hat{W}(\cl T \msZ, s), \quad \cl T \in \cl G \] for $\cl G = \cl G_0, \cl G_1$, and $\cl G_2$, respectively. To compute integrity bases for each $\cl G_i$, we define
\begin{gather}
    \msx_1 := \msu_\perp, \quad \msx_2 := \msv_\perp, \quad \msx_3 := (\msu_s)_\perp, \quad \msx_4 := (\msv_s)_\perp\\
    g_{ij} := \msx_i \cdot \msx_j, \quad a_{ij} := \mse_3 \cdot (\msx_i \times \msx_j) \\
    z_1 := u_3, \quad z_2 := v_3, \quad z_3 := u_{3,s}, \quad z_4 := v_{3, s} 
\end{gather}
where \(\msa_\perp := \msa - (\msa \cdot \bs e_3)\bs e_3\) for any \(\msa \in \mathbb R^3\). For ease of notation, we define the auxiliary sets 
\begin{gather}
     \cl E_1 := \{g_{11}, g_{22}, g_{33}, g_{44}, g_{13}, g_{24}, a_{12}, a_{23}, a_{34}, a_{14}, z_2, z_4\}\\
     \cl E_2 := \{g_{11}, g_{22}, g_{33}, g_{44}, g_{12}, g_{34}, a_{13}, a_{14}, a_{23}, a_{24}, z_1, z_2\}\\
      \cl F_1 := \{g_{12}, g_{23}, g_{34}, g_{14}, a_{13}, a_{24}, z_1, z_3 \}\\ 
      \cl F_2 := \{g_{13}, g_{14}, g_{23}, g_{24}, a_{12}, a_{34}, z_3, z_4\}
\end{gather}
With this notation, we can state the following, which is almost a complete consequence of classical invariant theory. 
\begin{prop}
\label{thm:invariant_representation}
  For each $i = 0,1,2$, $\cl G_i$ has an integrity basis $\cl P_i$ given by 
  \begin{gather}
      \cl P_0 = \{g_{ij}, 1 \leq i \leq j \leq 4\} \cup \{a_{ij}, 1 \leq i < j \leq 4\} \cup \{z_1, z_2, z_3, z_4\}\\
      \cl P_1 =\cl E_1 \cup \{ pq, p,q \in \cl F_1\}\\
      \cl P_2 =\cl E_2 \cup \{ pq, p,q \in \cl F_2\}.
  \end{gather}
  Thus, for each $s \in [0,L]$, if $\hat{W}(\cdot,s)$ is invariant under $\cl G_i$, there exists a smooth function $\Phi_{i;s} \in C^\infty(\bbR^k)$, such that \[\hat{W}(\msZ, s) = \Phi_{i;s}(\rho_i(\msZ))\] where $\rho_i: \bbR^n \to \bbR^k$ is a basis map for $\cl G_i$. 
\end{prop}
\begin{proof} 
	The case of $\cl G = \cl G_0$ is an immediate consequence of the First Fundamental Theorem of classical invariant theory (Theorem 11.2.1 in \cite{Procesi2007} or \cite{Weyl1946}).  For transverse isotropy, set $\cl G = \cl G_1$, and observe that {$\cl P_0 = \cl E_1 \cup \cl F_1$}. If $P \in \cl P[\bbR^n]^{\cl G}$, then $P$ is invariant under $\cl G_0$, so that {$P = Q(\cl E_1, \cl F_1)$} for some polynomial $Q$. Because of the additional invariance of $P$ under $\sigma_{\text{iso}}(\msE)$, we also have that {$P = Q(\cl E_1, -\cl F_1)$}. We can therefore write $P$ as the average \begin{align}
		\label{eq:average}
		P  = \frac{1}{2}\bigl(Q(\cl E_1, \cl F_1) + Q(\cl E_1, -\cl F_1)\bigr). 
	\end{align} But the right hand side contains no monomials of odd total degree in the elements of $\cl F$, which proves the representation for $\cl P_1$. The representation for $\cl P_2$ is proved in the same way as for $\cl P_1$. 
\end{proof}
\begin{rem}
   The integrity basis $\cl P_1$ is generally redundant, and some of its generators may be eliminated using the algebraic relations (often called \emph{syzygies} in the literature) among the $g_{ij}$ and $a_{ij}$. We retain the above form because it makes the parity imposed by transverse isotropy transparent. Similar procedures can be used to determine representations for ``prismatic'' bars \cite{Healy2002}, whose symmetries form a dihedral group action. 
\end{rem}

It is not always true that the energy density $\hat{W}$ will be smooth on all of $\bbR^n$. Consider, for example, 

\begin{align}
    \hat{W}(\msu, \msv, \msu_s, \msv_s, s) &= \frac{1}{2}\xi(s)u_3^2 + \frac{1}{2}\zeta(s)u_{3,s}^2\\
    & + K_EK_F \left[ \frac{v_3^2 -2}{v_3^2 (K_E - 4K_F) -2K_E} |\msu_\perp|^2 + \frac{v_3^2}{(2-v_3^2)[v_3^2(K_E-4K_F) + 8K_F]}v^2_{3,s}\right] \\
\end{align}
where $K_E, K_F$ are positive constants. 
{When restricted to planar, untwisted rod deformations, this energy reduces to the energy derived and analyzed in 
 equation (24) of \cite{Barchiesi2019a} for a special class of pantographic beams. Here, we assume that $0<v_3 < \sqrt{2}$, where $v_3=\sqrt{2}$ represents both the maximal stretching admitted by the pantographic lattice and the threshold at which the rod transitions to a uniformly extensible cable. If we consider $\hat{W}$ to be defined on all of $\bbR^n$, observe that the set $K = \{|v_3| \geq \sqrt{2}\} \cup \{\bs 0\}$ is $\cl G$ invariant for all three groups above. Inspired by this example, we derive a corollary of Proposition \ref{thm:invariant_representation} for functions which blow up or are not defined on a $\cl G$-invariant closed subset of the domain. 
 The proof relies on tools and results in \cite{Schwarz1975} and basic facts about partitions of unity. }
\begin{cor}
    Suppose that $K$ is a closed, $\cl G$-invariant subset of $\bbR^n$ and that $f \in C^\infty(\bbR^n \setminus K)^{\cl G}$ and $\rho$ is a $\cl G$ basis map. Then there exists a function $F \in C^\infty (\bbR^k \setminus \rho (K))$ such that 
    \begin{align}
        f(\msZ) = F(\rho(\msZ)), \quad \msZ \in \bbR^n \setminus K.
    \end{align}
\end{cor}
\begin{proof}
   Set $\Omega = \bbR^n \setminus K$. The map $\rho$ is proper: since $|\msZ|^2$ is an invariant polynomial, it is a polynomial in the generators defining $\rho$, so the inverse image under $\rho$ of a bounded set is bounded; the inverse image of a compact set is therefore closed and bounded, hence compact. Consequently, $\rho(K)$ is closed. Moreover, the basis map separates $\cl G$-orbits \cite{Schwarz1975}: if $\rho(\msX) = \rho(\msY)$, then $\msX$ and $\msY$ belong to the same orbit of $\cl G$. This implies that $\rho^{-1}(\rho (K)) = K$.
  
   Set $U = \bbR^k \setminus \rho(K)$. Since $U$ is an open set of $\bbR^k$, choose a locally finite open cover $\{U_i\}_{i \in \bbN}$ of $U$ such that $\overline{U}_i \subset U$ is compact for each $i$, and let $\{\varphi_i\}$ be a smooth partition of unity subordinate to this cover. 
   
   Define $\psi_i(\msX) := \varphi_i(\rho(\msX))$ and {$V_i := \rho^{-1}(U_i)$}. Then {$\{V_i\}$ is an open cover of $\Omega$, and $\overline{V}_i \subset \Omega$}. {Moreover,} $\text{supp } \psi_i \subset V_i$. Furthermore, $\{\psi_i\}$ is a partition of unity of $\Omega$, as 
   \begin{align}
       \sum_{i}\psi_i(\msX) = \sum_{i} \varphi_i(\rho(\msX)) \equiv 1, \quad \msX \in \Omega
   \end{align}
   and the sum remains locally finite. Since $\psi_i$ is $\cl G$-invariant and has support away from $K$, we have that $\psi_i f \in C^\infty(\bbR^n)$ and is $\cl G$-invariant. Thus, there exists an $F_i \in C^\infty(\bbR^k)$ such that $\psi_i f (\msX) = F_i(\rho(\msX))$. Let $\chi_i \in C^\infty(\bbR^k)$ be a smooth cutoff such that $\chi_i = 1$ on $\text{supp } \varphi_i$ and $\text{supp } \chi_i \subset U_i$. Then we still have $\psi_if = (\chi_i F_i) \circ \rho$, and can now decompose 
   \begin{align}
       f(\msX) &= \sum_{i} \psi_i(\msX) f(\msX) \\
            &= \sum_i \chi_i(\rho(\msX)) F_i(\rho(\msX)) \\
            &= F(\rho(\msX)),
   \end{align}
     for $\msX \in \Omega$, where $F(\msx) = \sum_{i} \chi_i(\msx)F_i(\msx)$. As a locally finite sum of $C^\infty(U)$ functions, we must have $F \in C^\infty(U)$. 
\end{proof}
 
\section{Buckling of a planar inextensible strain-gradient elastic Cosserat rod}

In this section, we study the planar buckling of an inextensible, unshearable, untwisted, transverse isotropic, strain-gradient Cosserat rod. 
This problem is a natural strain-gradient analogue of the classical Euler elastica and provides the simplest setting in which to study the intrinsic theory developed above. 

In classical Cosserat rod theory, related buckling problems are governed by second-order nonlinear equations and may be analyzed using a global bifurcation framework \cite{Rabinowitz1971, Rabinowitz1973, crandallrabinowitz71, antman05}. In the present setting, the dependence of the stored energy on strain gradients leads instead to a fourth-order nonlinear boundary value problem. We consider a cantilever rod, fixed at one end and subjected to a compressive force at the other. After rewriting the equilibrium equations in an integral form, we apply global bifurcation theory to prove the existence of nonlinear branches of equilibria bifurcating from the trivial straight configuration. The structure of the argument is related to Wolfe's analysis {\cite{wolfe90}} of magnetic rods, but the governing equations, constitutive assumptions, and natural boundary conditions arise here from the strain-gradient nature of the rod rather than from magnetic loading of a classical rod.

This analysis also clarifies the relation between linearized critical-load calculations and nonlinear postbuckling behavior. Eremeyev \cite{Eremeyev2026} studied critical forces and perturbative postbuckling expansions for third-gradient beam models. The results below show, in the present strain-gradient Cosserat setting, that each critical load of the linearized problem is an actual bifurcation point of the nonlinear equilibrium problem and that the corresponding local branch is analytic. Thus the perturbative expansion near a critical load is supported by a genuine local solution curve. Accordingly, our nonlinear bifurcation result provides a rigorous basis for interpreting the critical loads predicted by Eremeyev's linearized analysis as genuine buckling thresholds.

\subsection{Setup} Fix $L>0$ and let $I=(0,L)$. Let $\{\bs e_1, \bs e_2, \bs e_3\}$ denote the standard basis of $\bbR^3$.
 We suppose that the rod has length $L$ and reference configuration $\bs r_0(s) = s\bs e_3$ for $s\in [0,L]$. The $\{\bs e_3, \bs e_1 \}$-planar configuration of an inextensible and unshearable rod takes the form 
\begin{align}
    \label{eq:inextensible_unshearable_motion}
    \bs v(s) = \bs r_s (s) =\cos(\theta(s))\bs e_3 + \sin(\theta(s)) \bs e_1,
\end{align}
where $\theta(s)$ denotes the angle that $\bs r_s$ makes with $\bs e_3$. {Because the rod is untwisted, we take} $\bs d_2 = \bs e_2, \bs d_3 = \bs r_s$ and $\bs d_1 = \bs d_2 \times \bs d_3$, so that
\begin{align}
    \bs R = \bs d_k \otimes \bs e_k. 
\end{align}
It follows from these definitions that $\bs u (s) = \theta'(s) \bs e_2$, and therefore $\msu = {\bs R}^T\bs u = \theta'(s) \bs e_2$.  
 
To incorporate strain-gradient effects, we assume the stored energy density of the inextensible rod is given by \eqref{eq:stored_energy} with $A, B, \xi, \zeta \in C^2(\overline{I})$, $A \geq 0$ and $B \geq b >0$ for some $b  > 0$. Hyperelasticity guarantees that, in general,
\begin{gather}
    \msm(s) = A(s)\msu_\perp(s) + \xi(s) u_3(s) \bs e_3\\
    \msg(s)= B(s)(\msu_s)_\perp (s) + \zeta(s)u_{3,s}(s) \bs e_3.
\end{gather}
 In our specific context, $u_3 = u_{3,s} = 0$, so $\msu_\perp = \msu$, and therefore
\begin{align}
	\label{eq:constitutive_response_1}
	\bs m(s) =  A(s) \theta'(s) \bs e_2\\
	\bs g(s) = B(s) \theta''(s) \bs e_2. \label{eq:constitutive_response_2}
\end{align}
These constitutive relations generalize those of Euler's elastica discussed in \cite{antman05}. 

We assume that there are no body forces or body couples on the rod, i.e., $\bs f \equiv \bs \ell \equiv \bs 0$. For simplicity, we select the representative for which the reactive double-force vanishes, $\bs h\equiv\bs0$. As shown in Section~3.4, the constrained principle of virtual work uniquely determines the multiplier functional, and hence the gauge-invariant contact-force resultant $\bs F=\bs n-\mathscr D_s\bs h$, but it does not determine $\bs n$ and $\bs h$ separately. Thus, the condition $\bs h=\bs0$ is an additional model choice, not a consequence of the inextensibility and unshearability constraint or of the uniqueness of the multiplier functional. The heuristic penalty argument of Section~4.2 provides a constitutive motivation for this choice: it is the zero-length-scale member of the family selected by the relation $\msh=\mathsf D\partial_s\msn$. We adopt this zero-length-scale model throughout the present section. It follows that the contact double force vanishes at every cut, $\bs N_i=\bs0$, and that the contact-force resultant reduces to $\bs F=\bs n$.

We impose essential and natural boundary conditions
\begin{gather}
    \bs r(0) = \bs 0,\quad  \bs R(0) = \bs I , \quad \bs R_s(0) = \bs 0, \label{eq:placement_boundary_conditions} \\
	-\lambda\bs e_3 = \bs n_L = (\bs n - (\p_s \bs h - \bs u \times \bs h))_{s =L}\\
	\bs 0 = \bs M_L = \bs g(L), \quad \bs 0 = \bs m_L = (\bs m - (\p_s \bs g - \bs u \times \bs g) - \bs v \times \bs h)|_{s =L} \label{eq:couple_boundary_conditions},
\end{gather}
where $\bs I = \bs e_k \otimes \bs e_k$ is the identity tensor.
 These conditions guarantee that one end of the rod is clamped, and the other undergoes a compressive contact force $-\lambda \bs e_3$ with no contact couples. In terms of $\theta$, \eqref{eq:placement_boundary_conditions} and \eqref{eq:couple_boundary_conditions} become 
\begin{gather}
\theta(0) =0, \quad  \theta'(0) =0 \label{eq:theta_boundary_conditions_1}\\ 
    \theta''(L) = 0, \quad [A(s)\theta'(s) - (B(s)\theta''(s))' ]|_{s =L} = 0, \label{eq:theta_boundary_conditions_2}
\end{gather}
respectively. In the case that $A \equiv 0$, configurations of the rod having constant curvature store no energy, and are examples of floppy modes discussed in \cite{rizzi26}. Our clamping condition $\theta'(0)=0$ eliminates such floppy modes in the analysis below.

The governing equations \eqref{eq:fieldequations} in this case reduce to the vector equation
\begin{align}
	\label{eq:reduced_problem}
	\p_s(\bs m - (\p_s \bs g - \bs u \times \bs g)) + \bs v \times (-\lambda \bs e_3) =0.
\end{align}
Substituting \eqref{eq:constitutive_response_1} and \eqref{eq:constitutive_response_2}  into \eqref{eq:reduced_problem}, and using the boundary conditions \eqref{eq:theta_boundary_conditions_1} and \eqref{eq:theta_boundary_conditions_2}, reduces the problem to
\begin{align}
	\label{eq:dimensional_BVP}
	\begin{cases}
		(B \theta'')'' - (A \theta')' - \lambda\sin(\theta) = 0, \quad s \in (0,L)\\
		\theta(0) = \theta'(0) = \theta''(L) = [A\theta' - (B \theta'')']_{s =L} = 0.
	\end{cases}
\end{align}
Introducing the auxiliary functions 
\begin{align}
	\label{eq:auxiliary_functions}
	u(s):= \theta(Ls),\quad  \beta(s) := B(Ls)/L^4, \quad \alpha(s) := A(Ls)/L^2
\end{align}
for $s \in (0,1)$, we see that $u$ satisfies the rescaled equation
\begin{align}
	\label{eq:BVP}
	\begin{cases}
		(\beta u'')'' - (\alpha u')' - \lambda \sin u=0, \quad s \in (0,1)\\
	u(0) = u'(0) =0, \quad u''(1) = [\alpha u' - (\beta u'')']_{s =1} =0.
	\end{cases}
\end{align}
The \emph{linearized} boundary value problem associated with \eqref{eq:BVP} is given by
\begin{align}
	\label{eq:linearized_BVP}
	\begin{cases}
		(\beta u'')'' - (\alpha u')' - \lambda u =0, \quad s \in (0,1)\\
	u(0) = u'(0) =0, \quad u''(1) = [\alpha u' - (\beta u'')']_{s =1} =0.
	\end{cases}
\end{align}

Of course, $u \equiv 0$ on $[0,1]$ is a solution to \eqref{eq:BVP}. Our goal is to study non-trivial solutions to the nonlinear eigenvalue problem \eqref{eq:BVP}.

\subsection{Existence and simplicity of eigenvalues}
Let
\[
    V
    :=
    \left\{
        u\in H^2(0,1):
        u(0)=u'(0)=0
    \right\},
\]
and define the symmetric bilinear form
\[
    a(u,v)
    :=
    \int_0^1
    \left[
        \beta(s)u''(s)v''(s)
        +
        \alpha(s)u'(s)v'(s)
    \right]\,ds,
    \quad u,v\in V.
\]
Since $\beta\geq b>0$ and $\alpha\geq0$, the form $a$ is continuous
and coercive on $V$. Moreover, the embedding
$V\hookrightarrow L^2(0,1)$ is compact. The weak eigenvalue problem
associated with \eqref{eq:linearized_BVP} is therefore
\begin{equation}
\label{eq:weak-eigenvalue-problem}
    a(u,v)
    =
    \lambda\int_0^1u(s)v(s)\,ds
    \quad\text{for every }v\in V.
\end{equation}
The boundary conditions at $s=1$ arise as the natural boundary
conditions associated with this weak formulation. Standard spectral theory for compact self-adjoint operators \cite{Brezis2011} implies
that there exists a sequence of real eigenvalues
\[
    \lambda_1\leq\lambda_2\leq\cdots
\]
tending to infinity whose corresponding eigenfunctions form an orthonormal basis of
$L^2(0,1)$. Under the regularity assumptions on $\alpha$ and $\beta$,
the weak eigenfunctions are classical solutions of
\eqref{eq:linearized_BVP}. Taking $v=u$ in \eqref{eq:weak-eigenvalue-problem} gives
\[
    \lambda\int_0^1u^2\,ds
    =
    \int_0^1
    \left[
        \beta(u'')^2+\alpha(u')^2
    \right]\,ds.
\]
This identity guarantees that every eigenvalue is positive.

We say that an eigenvalue $\lambda_i$ of \eqref{eq:linearized_BVP} is \emph{simple} if the eigenspace $E_i$ associated to $\lambda_i$ is one-dimensional. Using the techniques of {\cite{Leighton1958}}, we can show  
\begin{prop}
	\label{thm:eigenvalue_simplicity}
	Suppose that $\alpha \in C^1([0,1])$ and $\beta \in C^2([0,1])$. Furthermore, suppose that $\alpha(s) \geq 0$ and $\beta(s) > 0$ for all $s \in [0,1]$. Then the eigenvalues of \eqref{eq:linearized_BVP} are simple. 
\end{prop}
\begin{proof}
	Fix an eigenvalue $\lambda =\lambda_i$ of \eqref{eq:linearized_BVP}, and let $E =  E_i$ denote the corresponding eigenspace. Suppose for contradiction that $\dim(E) > 1$, and let $u_1, u_2 \in E$ be linearly independent. For $u \in E$, we see upon integrating that
\begin{align}
	\label{eq:I_geq_0_1_integral}
	(\beta u'')' - \alpha u' = [(\beta u'')' - \alpha u']_{s=0} + \int_0^s \lambda u(t) dt.
\end{align}
We integrate once more to obtain 
\begin{align}
	\label{eq:I_geq_0_2_integral}
	\beta u'' &= [\beta u'']_{s=0} + \int_0^s [\alpha u'](\xi) + [(\beta u'')' - \alpha u']_{s=0} d\xi + \int_0^s (s-t) (\lambda u(t))dt  \\
		  &= [\beta u'']_{s=0} + s[(\beta u'')' - \alpha u']_{s=0} + \int_0^s [\alpha u'](\xi) d\xi + \int_0^s (s-t)\lambda u(t) dt
\end{align}
where we have used Fubini's theorem to switch the order of integration.  

Set $w := c_1u_1 + c_2u_2$ for $c_1, c_2 \in \bbR$ to be chosen presently. Then $w \in  E$, and in particular \[w(0) =w'(0) = w''(1) = [\alpha w' - (\beta w'')']_{s=1}= 0. \] Furthermore, we can choose $c_1$ and $c_2$ such that $c_1^2 +c_2^2 > 0$, $w''(0) = 0$ and 
\begin{align}
	\label{eq:w_boundary_inequality}
	[ (\beta w'')' -\alpha w']_{s=0} \geq 0. 	
\end{align}

Since we assumed $u_1$ and $u_2$ linearly independent, $w$ cannot be identically zero. By the uniqueness of the associated initial value problem for \eqref{eq:linearized_BVP}, the inequality \eqref{eq:w_boundary_inequality} must be strict.  

We claim that $w, w'$ and $w''$ are all strictly positive on $(0,1]$. Indeed, in the present case, we have $[(\beta w'')' - \alpha w']_{s=0} = (\beta w''')_{s=0}$, so that 
\begin{align}
	\label{eq:w_reinforcement_inequality}
	\beta w'' = s[\beta w''']_{s=0} + \int_0^s [\alpha w'](\xi) d\xi + \int_0^s (s-t)(\lambda w(t))dt
\end{align}
and $(\beta w''')_{s=0} > 0$. Since $\beta $ is positive, we must have  $w'''(0) >0$. Since $w(0) = w'(0) = w''(0) =0$, $w$ and $w'$ must be positive on an interval $(0, \eps)$ for some  $\eps>0$. If $w'$ were to vanish on (0,1], relabel $\eps$  as the first point at which this occurs. By \eqref{eq:w_reinforcement_inequality}, $w''(s) > 0$ on $(0, \eps]$, so that $w'$ is increasing and $w'(s) > 0$ on $(0, \eps]$, contradicting the definition of $\eps$. 

Thus, $w'$, and therefore $w$, are positive on $(0, 1]$. The identity \eqref{eq:w_reinforcement_inequality} implies that  $w''$ is positive on $(0, 1]$. However, this contradicts the fact that  $w''(1) =0$. Thus $E$ is one-dimensional. 
\end{proof}

\subsection{Bifurcation analysis}
Let $\cl X = C([0,1]; \bbR^4)$ with the standard norm. We can reformulate both \eqref{eq:BVP} and \eqref{eq:linearized_BVP} as first-order systems of integral equations for $\bs w = (w_0, w_1, w_2, w_3) \in \cl X$ that satisfy
\begin{gather}
	\label{eq:integral_system}
	w_0 = \int_0^s w_1 dt, \quad w_1 = \int_0^s \frac{w_2}{\beta} dt, \quad w_2 = \int_1^s (w_3 + \alpha w_1) dt,\\ 
	w_3 =  \int_1^s \lambda \sin(w_0) dt  
\end{gather}
and 
\begin{gather}
\label{eq:linear_integral_system}
	w_0 = \int_0^s w_1 dt, \quad w_1 = \int_0^s \frac{w_2}{\beta} dt, \quad w_2 = \int_1^s (w_3 + \alpha w_1) dt, \\
	w_3 =  \int_1^s \lambda w_0 dt, 
\end{gather}
respectively. This formulation naturally incorporates the above boundary conditions. We define the operator $\bs f: \bbR \times \cl X \to \cl X$ by 
\begin{align}
	\label{eq:bifurcation_operator}
	\bs f[\lambda, \bs w]& := (\bs I -\bs J) \bs w  - \lambda \bs N \bs w - \bs h(\lambda, \bs w)
\end{align}
where $\bs I: \cl X \to \cl X$ is the identity operator, and  
\begin{gather}
	\label{eq:linearized_bifurcation}
	\bs J \bs w := \left(\int_0^s w_1 dt, \int_0^s \frac{w_2}{\beta } dt, \int_1^s (w_3 + \alpha w_1) dt, 0 \right),\\
	\bs N\bs w := \left(0, 0, 0, \int_1^s  w_0 dt\right), \\
     \bs h(\lambda,\bs w) := \left(0, 0, 0, \int_{1}^s \lambda (\sin w_0 - w_0 ) dt \right). 
\end{gather}
The operators $\bs J, \bs N: \cl X \to \cl X$ and $\bs h: \bbR \times \cl X \to \cl X$ are compact by an application of the Arzel\`a--Ascoli theorem, and $\bs h(\lambda, \bs w) = o(\|\bs w\|)$ as $\bs w \to \bs 0$ uniformly on bounded $\lambda$ intervals. We say that $\mu$ is a characteristic value of $\bs f$ if there exists a nontrivial $\bs w$ satisfying $\bs 0 = (\bs I - \bs J)\bs w - \mu \bs N \bs w$. We let $\cl E (\bs f)$ denote the set of characteristic values of $\bs f$. We note that $\cl E (\bs f)$ is precisely the set of eigenvalues of the linearized boundary value problem \eqref{eq:linearized_BVP}. 

\begin{prop}
	\label{prop:fredholm_alternative}
	The operator $\bs I - \bs J : \cl X \to \cl X$ has a bounded inverse.
\end{prop}
\begin{proof}
We claim that the equation $\bs J \bs w = \bs w$ has no non-trivial solutions, so the result follows from the Fredholm alternative since $\bs J$ is compact \cite{Brezis2011}. Indeed, suppose $\bs w$ solves $\bs J \bs w= \bs w$. Then we must have $w_0 \in C^3[0,1]$ and 
\begin{align}
	\label{eq:w_0_third_order}
	\begin{cases}
	(\beta w_0'')' - \alpha w_0' = 0 \\
	w_0(0) = w_0'(0) = w_0''(1) = 0. 
	\end{cases} 
\end{align}
Setting $v = w_0'$ we find that, via integration by parts and the fact that $v(0) = v'(1) =0$,
\begin{align}
	\label{eq:v_iBP}
	\int_0^1 \beta (v')^2 ds  + \int_0^1 \alpha v^2 ds= 0.	
\end{align}
Since $\beta >0$ and $\alpha \geq 0$, we must have $v' = 0$ on $[0,1]$. From the boundary condition $v(0) =0, w_0(0)= 0$, we obtain $w_0 =0$ on $[0,1]$. This implies that $\bs w = \bs 0$. 
\end{proof}
Using Proposition \ref{prop:fredholm_alternative}, we may define 
\begin{align}
	\label{eq:reformulated_bifurcation}
	\bs F [\lambda, \bs w] :&= (\bs I - \bs J)^{-1}\bs f(\lambda, \bs w)\\
				&=  (\bs I - \lambda \bs L) \bs w - \bs g(\lambda, \bs w).
\end{align}
where 
\begin{align}
	\label{eq:reformulated_bifurcation_components}
	\bs L = (\bs I - \bs J)^{-1} \bs N, \quad \bs g = (\bs I - \bs J)^{-1} \bs h. 
\end{align}
Since $(\bs I - \bs J)^{-1}$ is bounded, $\bs L: \cl X \to \cl X$ and $\bs g: \bbR \times \cl X \to \cl X$ are compact operators, and $\bs g(\lambda, \bs w) = o(\|\bs w\|)$ as $\bs w \to 0$ for bounded $\lambda$ intervals. Thus the operator $\bs F: \bbR \times \cl X \to \cl X$ fits within the bifurcation framework of \cite{Rabinowitz1971, Rabinowitz1973}. We say $\mu$ is a characteristic value of $\bs F$ if there exists a nontrivial solution to $(\bs I - \mu \bs L) \bs w =0$, and denote the set of characteristic values of $\bs F$ as $\cl E(\bs F)$. We see that $\cl E (\bs F) = \cl E(\bs f)$, and henceforth denote both by $\cl E$.

For $\lambda \in \cl E$, we define, as in \cite{Rabinowitz1971, Rabinowitz1973}, the \emph{algebraic multiplicity} $m$ of $\lambda$ to be 
\begin{align}
	\label{eq:algebraic_multiplicity}
	m(\lambda, \bs L) :&= \dim \left[\bigcup_{n=1}^\infty \ker((\bs I - \lambda \bs L)^n)\right].
\end{align}
\begin{prop}
	\label{prop:algebraic_multiplicity_one}
	Under the same assumptions on $\alpha$ and $\beta$ as in Proposition \ref{thm:eigenvalue_simplicity}, $m(\lambda, \bs L) = 1$ for all $\lambda \in \cl E$. 
\end{prop}
\begin{proof}
	Fix $\lambda \in \cl E$. Recall that $\dim E(\lambda) = 1$, where $E(\lambda)$ denotes the eigenspace associated to $\lambda$ in \eqref{eq:linearized_BVP}. The integral formulation \eqref{eq:linear_integral_system} implies \begin{align}
		\label{eq:dim_f_F}
		\dim \ker(\bs I - \lambda \bs L) &= \dim \ker((\bs I - \bs J)^{-1}( \bs I -\bs J - \lambda \bs N))\\
						 &= \dim \ker (\bs I -\bs J -\lambda \bs N) =  1. 
	\end{align} Suppose there exists $\bs u \in \ker((\bs I - \lambda \bs L)^2)$ but $\bs u \notin \ker (\bs I - \lambda \bs L)$. Setting $\bs w := (\bs I - \lambda \bs L) \bs u \neq \bs 0$, we see that $\bs w \in \ker (\bs I -\lambda \bs L)$, i.e. $\bs 0 = (\bs I - \lambda \bs L)\bs w,$
	which implies that we also have 
	\begin{align}
	0 =  (\bs I - \bs J)( \bs I - \lambda \bs L)\bs w = (\bs I - \bs J - \lambda \bs N)\bs w. 
	\end{align}
	Setting $\bs w = (w_0, w_1, w_2, w_3)$, we see that $w_0$ solves \eqref{eq:linearized_BVP}. Observe also that $(\bs I - \bs J)\bs w = (\bs I - \bs J -\lambda \bs N)\bs u$, which implies 
	\begin{align}
		\label{eq:w-u0-relationship}
		(0, 0, 0, w_3) = \left( u_0 - \int_0^s u_1 dt, u_1 - \int_0^s \frac{u_2}{\beta}dt, u_2 - \int_1^s (u_3 + \alpha u_1)dt, u_3 - \lambda \int_1^s u_0 dt \right).
	\end{align}
	Since $w_3 = \lambda \int_1^s w_0 dt$, we use \eqref{eq:w-u0-relationship} to see that $u_0$ solves
	\begin{align}
		\label{eq:w_u_differential_relationship}
		\begin{cases} (\beta u_0'')'' - ( \alpha u_0')' -\lambda u_0 = \lambda w_0\\
				u_0(0) = u_0'(0) = u_0''(1) = [\alpha u_0' - (\beta u_0'')']_{s=1} = 0
		\end{cases}. 
	\end{align}
	Pairing \eqref{eq:w_u_differential_relationship} against $w_0$ and integration by parts yields 
	\begin{align}
		\label{eq:w_u_contradiction}
		\int_0^1 \lambda w_0^2 ds &= \int_{0}^1 (\beta u_0'')'' w_0 - (\alpha u_0')'w_0 - \lambda u_0 w_0 ds\\
					  &= \int_{0}^1 (\beta w_0'')'' u_0 - (\alpha w_0')'u_0 - \lambda w_0 u_0 ds = 0.  
	\end{align}
	This integration by parts is possible because $u_0$ and $w_0$ satisfy the same boundary conditions. Since $\lambda >0$, we must have that $w_0 \equiv 0$ on $[0,1]$, which implies that $\bs w = \bs 0$. This contradicts the fact that $\bs w \neq \bs 0$. Thus \[\ker((\bs I - \lambda \bs L)^2)  = \ker((\bs I - \lambda \bs L)), \] and this implies via induction that $m(\lambda, \bs L) = 1$.
\end{proof}
We are now prepared to state our main result. We define $\cl S(\bs f)$ to be the closure of non-trivial solutions to $\bs f[\lambda, \bs w] = \bs 0$ in $\bbR \times \cl X$, and $\cl S(\bs F)$ is defined similarly. Observe that $\cl S(\bs f) = \cl S(\bs F)$, so we denote both by $\cl S$.  We define a \emph{bifurcation point} $(p, \bs b) \in \bbR \times \cl X$ to be a point where $\bs f[p, \bs b] = \bs 0$ and, for any neighborhood of $(p, \bs b)$, there exists a non-trivial solution to $\bs f = \bs 0$. See \cite{antman05, Rabinowitz1973} for more details.   
\begin{thm}[Global bifurcation for cantilever strain-gradient rod]
\label{thm:global_bif_theorem}
	The set of bifurcation points from the trivial branch of $\bs f$ is exactly $\{(\lambda, \bs 0), \lambda \in \cl E\}$. For $\lambda \in \cl E$, $(\lambda, \bs 0)$ belongs to a maximal connected subset $\cl C(\lambda) \subset \cl S$ which satisfies at least one of the following:
	\begin{enumerate}
		\item $\cl C(\lambda)$ is unbounded in $\bbR \times \cl X$. 
		\item $\cl C(\lambda)$ contains another bifurcation point $(\mu, \bs 0)$. 
	\end{enumerate}
	Furthermore, near $(\lambda, \bs 0)$, $\cl C(\lambda)$ is locally an analytic curve. 
\end{thm}
\begin{proof}
	Each $\lambda \in \cl E$ has algebraic multiplicity $1$, so Theorem 1.3 of \cite{Rabinowitz1971} implies that $(\lambda, \bs 0)$ is a bifurcation point of $\bs F$, and therefore of $\bs f$. Theorem 1.3 also provides the global alternative for $\cl C(\lambda)$. If $\lambda \not\in \cl E$, the implicit-function theorem implies that $(\lambda, \bs 0)$ cannot be a bifurcation point of $\bs f$ (see Theorem 5.4.1 in \cite{antman05}). Observe that $\bs h: \bbR \times \cl X \to \cl X$ is analytic, which implies that $\bs g$ is analytic as well. The remark following Theorem 1.19 of \cite{Rabinowitz1973} and Theorem 1.7 of \cite{crandallrabinowitz71} then yield the local analyticity of $\cl C(\lambda)$. 
\end{proof}
{
\begin{rem}
    In the classical case, it is well known that only option $(1)$ in the above theorem can occur for the non-trivial branches. The proof of this fact relies on nodal properties of the associated second-order equation. The question of whether such a result is true in the present case is open.
\end{rem}
}
\subsection{A quantitative lower bound for the first bifurcation load}
For the remainder of this section, we
assume that $\alpha = 0$.
By the variational characterization of the first eigenvalue \cite{Brezis2011},
\[
    \lambda_1
    =
    \min_{u\in V\setminus\{0\}}
    \frac{a(u,u)}{\|u\|_{L^2}^2}.
\]
Consequently,
\[
\begin{aligned}
    \lambda_1
    &\geq
    \widehat{\beta}
    \min_{v\in V\setminus\{0\}}
    \frac{\displaystyle\int_0^1
        \bigl(v''(s)\bigr)^2\,ds}
         {\displaystyle\int_0^1 v(s)^2\,ds} 
    = \widehat{\beta}\,\mu_1^4, \quad \widehat{\beta} := \min_{s \in [0,1]} \beta,
\end{aligned}
\]
where $\mu_1^4$ is dimensionless and the first eigenvalue of the constant-coefficient boundary value
problem
\begin{equation}
\label{eq:constant-coefficient-linear-problem}
\begin{cases}
    w''''=\mu^4w, & s\in(0,1),\\
    w(0)=w'(0)=0,
    w''(1)=w'''(1)=0.
\end{cases}
\end{equation}
The same lower bound holds for every $\alpha\geq0$, because the nonnegative term $\int_0^1\alpha(u')^2\,ds$ may simply be discarded from the Rayleigh quotient. The comparison problem \eqref{eq:constant-coefficient-linear-problem} and its natural boundary conditions are unchanged. The characteristic values $\mu_k$ are the positive roots of
\[
    1+\cosh(\mu)\cos(\mu)=0.
\]
In particular,
\[
    \mu_1 \approx 1.875104069,
    \quad
    \mu_1^4 \approx 12.362363368,
\]
and therefore
\begin{equation}
\label{eq:first-load-lower-bound}
    \lambda_1
    \geq
    12.362363368\,\widehat{\beta}.
\end{equation} 
This value of $\mu_1$ was also computed in the purely linear analysis of \cite{Eremeyev2026}.  
Since the eigenvalues of the linearized problem are precisely the
bifurcation values established above, \eqref{eq:first-load-lower-bound}
gives a lower bound for the load at which bifurcation from the straight
configuration may occur. This statement concerns bifurcation from the
trivial branch; it does not exclude the possible existence of nontrivial
equilibria at smaller loads that are not connected to the trivial
configuration. For reference, Table~\ref{tab:cantilever-characteristic-values} compares
the first several eigenvalues
\eqref{eq:constant-coefficient-linear-problem} with those of the
corresponding classical cantilever problem. The two columns should not be read as a coefficient-free comparison of physical stiffness, since the classical and strain-gradient moduli have different dimensions. Rather, under the displayed dimensionless normalizations, the spectra exhibit the characteristic fourth-order and second-order growth rates $\mu_k^4\sim k^4$ and $\mu_{0k}^2\sim k^2$, respectively.

\begin{table}[ht]
\centering
\begin{tabular*}{0.8\textwidth}{@{\extracolsep{\fill}}ccc}
\hline
$k$ & $\mu_k^4$ & $\mu_{0k}^2$ \\
\hline
1 & 12.362363368 & 2.467401100 \\
2 & 485.518818513 & 22.206609902 \\
3 & 3806.546266391 & 61.685027507 \\
4 & 14617.273305119 & 120.902653913 \\
5 & 39943.831778509 & 199.859489122 \\
6 & 89135.405071423 & 298.555533133 \\
7 & 173881.315656106 & 416.990785946 \\
8 & 308208.452093651 & 555.165247561 \\
9 & 508481.543266067 & 713.078917979 \\
10 & 793403.134540116 & 890.731797198 \\
\hline
\end{tabular*}
\caption{Comparison of the first ten eigenvalues $\mu^4_k$ and {$\mu_{0k}^2$} of the
        linearized strain-gradient and linearized classical cantilever problem, respectively.
        A comparison of $\mu_k$ and $\mu_{0k}$ can be found in
        \cite{Eremeyev2026}. Numerical values computed with Wolfram 15.0.}
\label{tab:cantilever-characteristic-values}
\end{table}

\subsection{Numerical illustration of buckled states}

We conclude with a numerical illustration of the nonlinear equilibria
whose existence is guaranteed by the bifurcation result above. We
specialize further by assuming that
\[
    \beta(s)\equiv\beta_0
\]
for some constant $\beta_0>0$. The nonlinear boundary-value problem then
becomes
\[
\begin{cases}
    \beta_0 u''''-\lambda\sin u=0,
        & s\in(0,1),\\
    u(0)=u'(0)=0,
    u''(1)=u'''(1)=0.
\end{cases}
\]
Introducing the dimensionless load
\[
    \widehat{\lambda}:=\frac{\lambda}{\beta_0},
\]
we obtain
\begin{equation}
\label{eq:dimensionless-nonlinear-problem}
\begin{cases}
    u''''-\widehat{\lambda}\sin u=0,
        & s\in(0,1),\\
    u(0)=u'(0)=0,
    u''(1)=u'''(1)=0.
\end{cases}
\end{equation}
This is the Euler--Lagrange equation associated with the functional
\begin{equation}
\label{eq:dimensionless-energy}
    \Phi_{\widehat{\lambda}}(u)
    =
    \int_0^1
    \left[
        \frac12\bigl(u''(s)\bigr)^2
        +
        \widehat{\lambda}\bigl(\cos u(s)-1\bigr)
    \right]ds
\end{equation}
on the space $V$. The conditions at $s=0$ are imposed as essential
boundary conditions, while the conditions at $s=1$ arise naturally from
the variational principle. 

The functional \eqref{eq:dimensionless-energy} was discretized using
cubic Hermite finite elements for the tangent angle $u$. The linearized
finite-element problem was first used to approximate the eigenvalues and
eigenmodes associated with the first three buckling modes. Each
eigenmode was then used to initialize a nearby nontrivial equilibrium,
and the corresponding positive branch was followed using
pseudo-arclength continuation. Continuation was performed using \texttt{BifurcationKit.jl} \cite{BifurcationKit}. For each computed angle field $u$, the rod centerline was reconstructed
from
\begin{equation}
\label{eq:centerline-reconstruction}
    x(s)=\int_0^s\cos\bigl(u(t)\bigr)\,dt,
    \quad
    y(s)=\int_0^s\sin\bigl(u(t)\bigr)\,dt.
\end{equation}
Representative configurations along the
three branches are displayed in
Figure~\ref{fig:strain-gradient-hermite-modes}.

The purpose of these computations is only to illustrate the geometric
form of the nonlinear equilibria predicted by the preceding analysis.
A detailed numerical study of the branches, including questions of
convergence, stability, and secondary bifurcation, is left for future
work.

\begin{figure}[p]
    \centering
    \captionsetup{font=small}
    \captionsetup[subfigure]{font=small,skip=3pt}
    \begin{subfigure}[t]{0.69\linewidth}
        \centering
        \includegraphics[width=\linewidth,trim=0 35 6 48,clip]{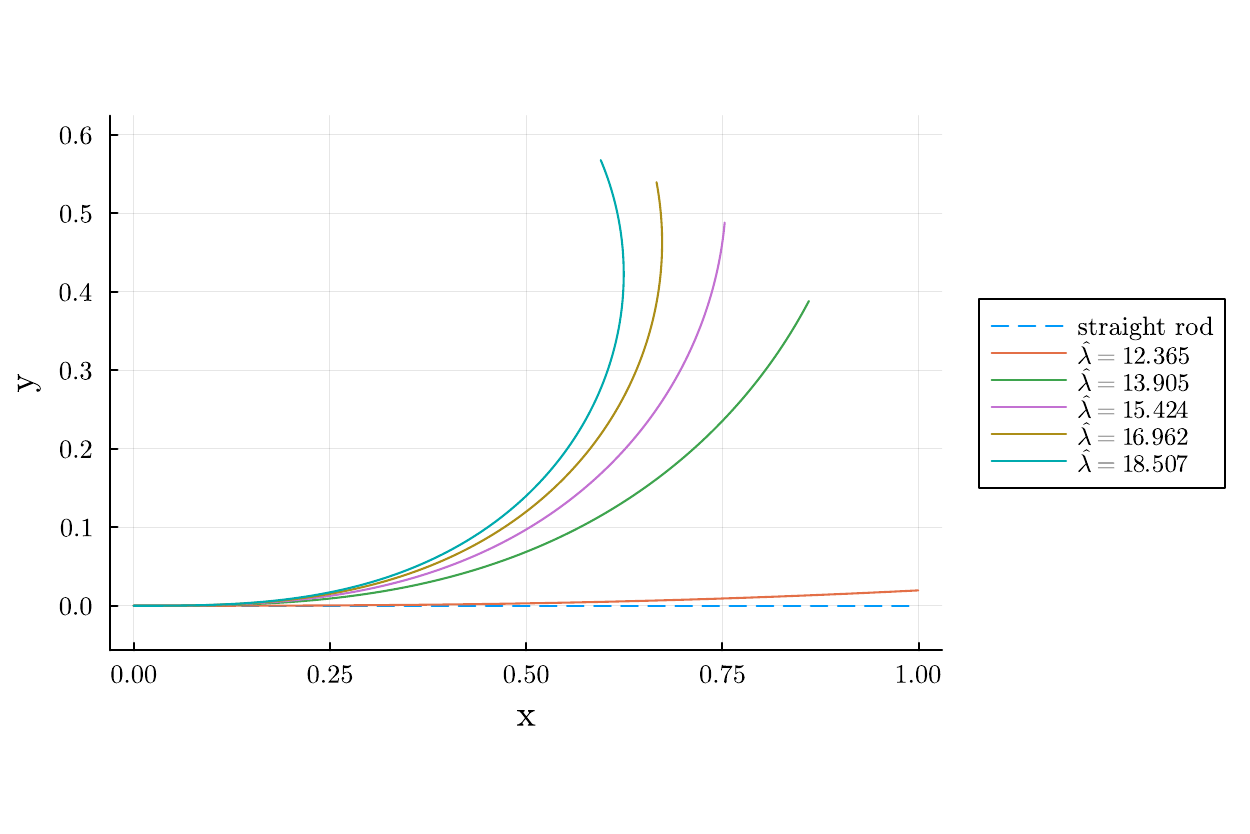}
        \caption{Branch 1}
        \label{fig:strain-gradient-hermite-mode-1}
    \end{subfigure}\par\smallskip
    \begin{subfigure}[t]{0.69\linewidth}
        \centering
        \includegraphics[width=\linewidth,trim=0 35 6 48,clip]{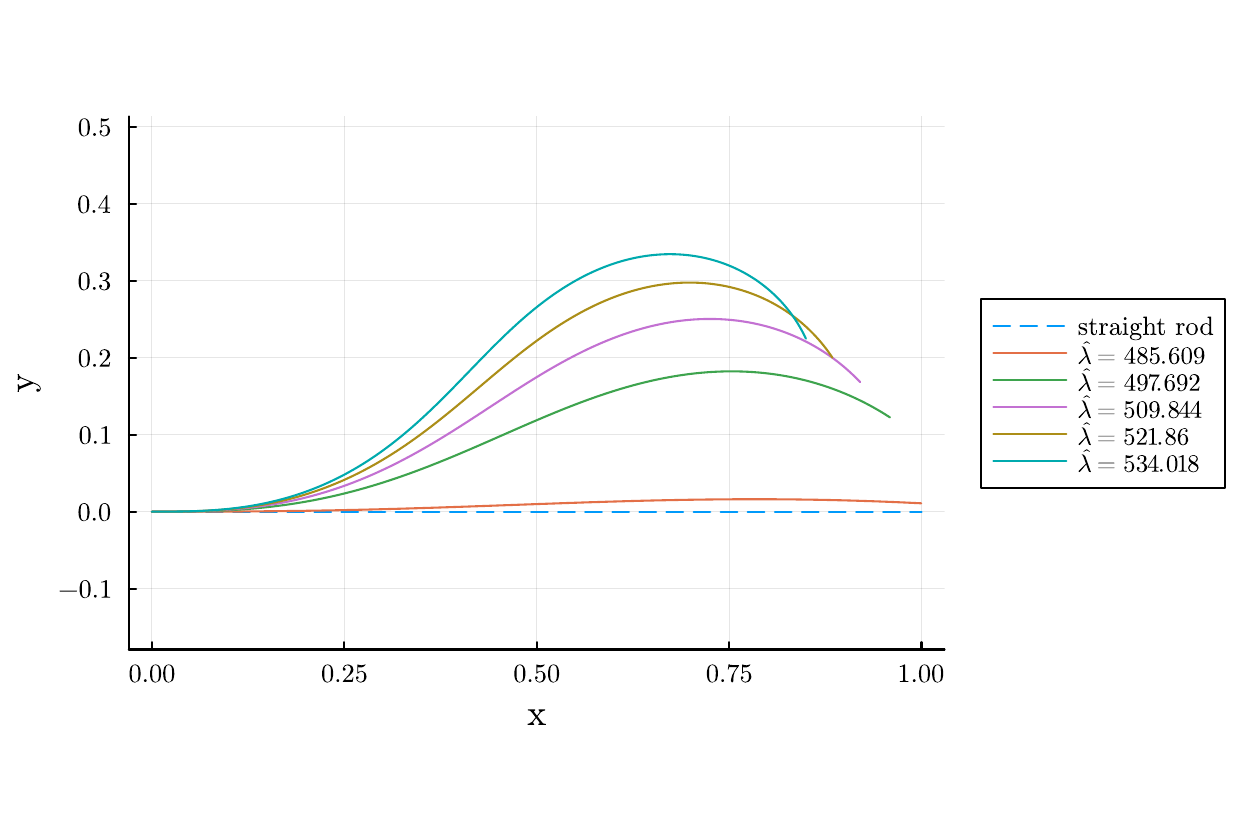}
        \caption{Branch 2}
        \label{fig:strain-gradient-hermite-mode-2}
    \end{subfigure}\par\smallskip
    \begin{subfigure}[t]{0.69\linewidth}
        \centering
        \includegraphics[width=\linewidth,trim=0 35 6 48,clip]{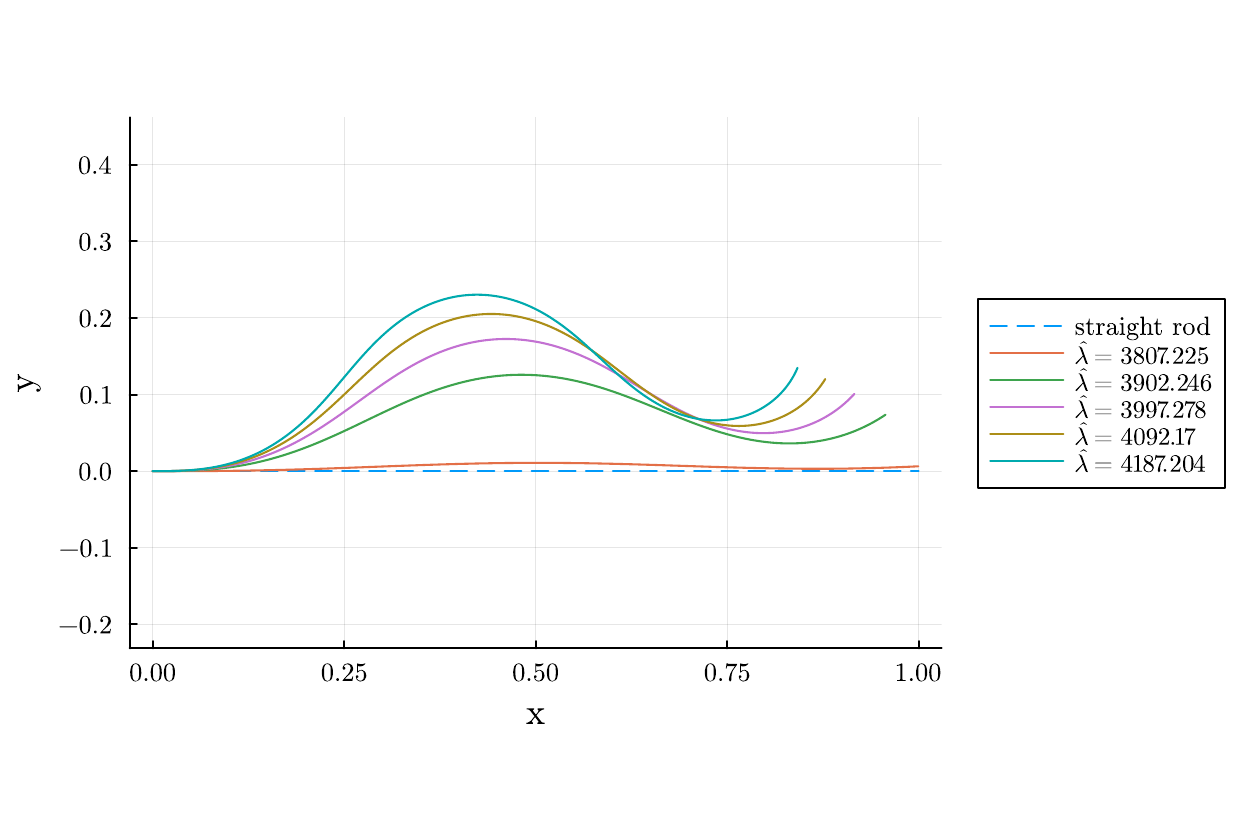}
        \caption{Branch 3}
        \label{fig:strain-gradient-hermite-mode-3}
    \end{subfigure}

    \caption{Representative equilibria on the positive branches issuing
        from the first three bifurcation loads of
        \eqref{eq:dimensionless-nonlinear-problem}, computed as described
        in the text. Legends give the dimensionless loads
        $\widehat{\lambda}$. Each panel shows five equilibria selected at
        evenly spaced continuation parameter values over load ranges
        equal to 50\%, 10\%, and 10\% of the corresponding initial
        bifurcation load for branches 1, 2, and 3, respectively.
        No conclusion regarding stability is intended.}
    \label{fig:strain-gradient-hermite-modes}
\end{figure}

\FloatBarrier
\section{Conclusion}

We have developed a direct theory of strain-gradient special Cosserat rods based on the principle of virtual work, without prescribing a three-dimensional parent material or assuming the existence of a stored energy. The formulation yields governing equations, higher-order contact interactions, and natural boundary conditions, and accommodates internal constraints, objective constitutive laws, and material symmetries. For a planar cantilever under terminal compression, we establish that the eigenvalues of the linearized problem are simple and are precisely the nonlinear bifurcation loads, each giving rise to a connected global continuum of equilibria.

The theory developed here provides a foundation for further studies of nonlinear dynamics and global postbuckling behavior. On the dynamical side, a natural extension is to allow the kinetic-energy density to depend on material-coordinate derivatives of the centerline velocity and cross-sectional angular velocity, and to determine the corresponding inertial body forces, field equations and boundary interactions. This extension raises the broader question of which constitutive and boundary assumptions ensure well-posed dynamics, particularly when the stored energy admits floppy modes. For infinite rods, a related question is the existence of solitary waves and, if so, whether these waves are stable. On the equilibrium side, it remains open whether every global continuum of buckled states for the cantilever beam is unbounded or can connect distinct bifurcation points. Resolving these questions would connect the present theory to the propagation of localized disturbances and the accessibility of large-deformation equilibria in synthesized higher-gradient one-dimensional materials, including pantographic beams. 

\subsection*{Conflict of Interest Statement.} The authors declare that they have no conflicts of interest.

\subsection*{Data Availability Statement} No datasets were generated or analyzed during the course of this study.

\FloatBarrier
\bibliographystyle{abbrv}
\bibliography{literature}

\bigskip

\footnotesize

	\noindent \textsc{Department of Mathematics, The University of North Carolina at Chapel Hill, CB 3250 Phillips Hall,
		Chapel Hill, NC 27599}\\
	\noindent \textit{E-mail address}: \texttt{scott\_hallyburton@unc.edu}
	
	\bigskip	
	
	\noindent \textsc{Department of Mathematics, The University of North Carolina at Chapel Hill, CB 3250 Phillips Hall,
		Chapel Hill, NC 27599}\\
	\noindent \textit{E-mail address}: \texttt{crodrig@email.unc.edu}

\end{document}